\documentclass[twoside,11pt]{article}
\usepackage{latexsym}
\usepackage{epsfig}
\usepackage{amsfonts}
\usepackage{amsthm,amsmath,amsfonts,amssymb}
\numberwithin{equation}{section}

\newcommand{\supp}{\operatorname{supp}}

\newcommand{\opm}{\operatorname{op}}

\newcommand{\Aop}{\operatorname{\cal A}}
\newcommand{\Pop}{\operatorname{\cal P}}
\newcommand{\Gop}{\operatorname{\cal G}}

\newcommand{\Kop}{\operatorname{\cal K}}
\newcommand{\Iop}{\operatorname{\cal I}}
\newcommand{\Qocc}{\operatorname{Q_{\mbox{occ}}}}
\newcommand{\Qvirt}{\operatorname{Q_{\mbox{virt}}}}

\begin{document}
\newtheorem{assumption}{Assumption}
\newtheorem{proposition}{Proposition}
\newtheorem{definition}{Definition}
\newtheorem{lemma}{Lemma}
\newtheorem{theorem}{Theorem}
\newtheorem{observation}{Observation}
\newtheorem{remark}{Remark}
\newtheorem{corollary}{Corollary}
\newtheorem{example}{Example}

\title{Small distance behavior of one-particle Green's functions in electronic structure theory}

\author{Heinz-J{\"u}rgen Flad$^\dag$ and Michael Griebel$^{\dag \ddag}$
\ \\
$^\dag${\small Institut f\"ur Numerische Simulation, Universit\"at Bonn,
Friedrich-Hirzebruch-Allee 7 53115 Bonn}\\
$^\ddag${\small Fraunhofer Institut f\"ur Algorithmen und Wissenschaftliches Rechnen SCAI,}\\ 
{\small Schloss Birlinghoven 1, 53757 Sankt Augustin}
}

\maketitle

\begin{abstract}
Within the framework of many-particle perturbation theory, we develop an analytical approach that allows us to determine the small distance behavior of Green's functions and related quantities in electronic structure theory. As a case study, we consider the one-particle Green's function up to 2nd order in the perturbation approach. We derive explicit expressions for the leading order terms of the asymptotic small distance behavior. In particular, we demonstrate the appearance of a logarithmic term in the corresponding 2nd order Feynman diagrams.  Our asymptotic analysis leads to an improved classification scheme for the diagrams, which takes into account not only the perturbation order, but also the asymptotic smoothness properties near their diagonals. Such a classification may be useful in the design of numerical algorithms and helps to improve their efficiency.

\end{abstract}

\section{Introduction}
Computational methods in electronic structure theory are now at a mature stage, allowing simulations ranging from atoms and molecules to macromolecules and solids. Depending on the size of the system and the property under consideration, a large variety of many-particle models and sophisticated algorithms for their simulation are available, see e.g. \cite{Fulde,HJO} and the references cited therein. While a systematic control of the modeling errors seems illusory, a systematic control of the discretization errors should be mandatory. In practical applications, however, this topic is often poorly treated, for obvious reasons:
In general, many-particle models are highly nonlinear, and the discretization is done with respect to basis sets, such as the popular and successful Gaussian-type basis functions in quantum chemistry, which are neither stable nor systematically refinable from a mathematical point of view. Rather, they belong to the realm of approximate approximation theory \cite{MS07}. Moreover, most of the quantities of interest, such as many-particle wavefunctions, density matrices, Green's functions, and response functions, have a complicated singular structure due to the underlying singular Coulomb interactions, which makes it difficult to systematically improve the discretization error. With the exception of wavefunctions, which are eigenfunctions of multiparticle Hamilton operators, and the corresponding single-particle-densities, see e.g. \cite{FFHS20,FHO2S07,FHO2S09,FS21,Yserentant}, there is no asymptotic analysis of the underlying singular behavior of these quantities.   
As a consequence, rigorous results on convergence rates and error estimates are largely lacking in the literature. Even for systematic basis functions, such as plane waves, wavelets, or finite elements, the size of the basis set is often determined by the computational cost of the problem under consideration rather than by the requirement to achieve a certain prescribed size of discretization error. Except for pure mean-field models such as Hartree-Fock and density functional theory,\footnote{In these models, it is only the possibly singular Coulomb potentials of the nuclei that give rise to point-like conical singularities. Such types of singularities are efficiently treated by Gaussian-type basis functions and can be adaptively treated by systematically refined local basis functions such as finite elements or wavelets.} many-particle models generally suffer from edge singularities and higher-order corner singularities, which severely limit the convergence rates that can be achieved by systematic basis sets \cite{Yserentant}. Even the use of adaptive discretization schemes, e.g.~by best $N$-term approximation, does not overcome this obstacle \cite{FHS07}. The only promising approaches are either to incorporate emerging singular terms into the basis set or to remove singularities from the model by an appropriate subtraction scheme. The first approach is found for example in explicitly correlated many-particle methods, such as in coupled cluster theory, which include the leading order singular term of the wavefunction in the basis set. In practice, such an approach works surprisingly well and produces the most accurate results for molecules of moderate size \cite{KMTV06} that are currently available.
The second approach is related to renormalization schemes commonly used in quantum field theory.
The basic idea is to subtract singular asymptotic terms from the quantities under consideration, to determine them individually, and to sum them up in renormalized coupling constants or local corrections.
An example is range separation methods, where the singular long-range Coulomb potential is split into a singular short-range and a smooth long-range part.
For both approaches it is necessary to identify occurring singular parts before the actual simulation. This is a largely unexplored territory,  especially for time- or frequency-dependent quantities. 

The purpose of this work is to prepare the ground for a systematic asymptotic analysis of the small disctance behavior of time- and frequency-dependent quantities in electronic structure theory, such as the Green's functions and related response functions. This topic has a long tradition in quantum field theory and is related to issues such as ultra-violet divergencies and corresponding renormalization schemes. These make heavy use of space-time symmetries, such as translational invariance, which is reflected in the fact that most of the work is done in momentum space. In electronic structure theory, however, spatial symmetries are often broken by external potentials. Therefore, we rely instead on a pseudo-differential calculus, which can be seen as a compromise between localization in momentum and configuration space. 
The mathematical machinery we use to study the asymptotic behavior of Green's functions and related quantities  is  borrowed from spectral theory, microlocal analysis, and in particular singular analysis. 
Our primary tool is Schulze's singular pseudo-differential calculus and corresponding function spaces, which encompass the whole hierarchy of Coulomb singularities that can arise in many-particle Coulomb systems.
To illustrate our approach and the techniques involved, in this article we consider the asymptotic small distance behavior of one-particle Green's functions up to second order in perturbation theory.    

This paper is structured as follows: After some introductory remarks on Green's functions in quantum many-particle theory, we discuss in Section \ref{spectralresolution} our approach of translating back from formal spectral resolutions of operator resolvents to a formalism in the framework of pseudo-differential operator algebras. This is an unavoidable first step in our asymptotic analysis, due to the fact that the second quantization, commonly used in quantum many-particle theory, is based on formal spectral resolutions in terms of eigenfunctions of appropriately chosen one-particle Hamilton operators. It requires sophisticated techniques, such as the concept of generalized eigenfunctions, to  deal with these spectral resolutions within a rigorous mathematical framework. 
 
In the main part of the paper, we apply the aforementioned technique and study the small distance behavior of one-particle Green's functions for many-particle systems interacting via singular Coulomb interactions.
After a brief outline of some basic concepts and techniques of singular analysis in Section \ref{PsiDO}, we discuss how to apply Agmon's limiting absorption principle in the framework of singular analysis. In particular, this allows us to consider certain limits of complex resolvents required for Green's functions in quantum many-particle theory. This is followed by a detailed discussion of various types of Feynman diagrams from the point of view of singular analysis, which contribute to  one-particle Green's functions.
Our asymptotic analysis is explicitly carried out only up to second order in perturbation theory, but can be extended to higher orders with obvious modifications.
It should be noted that our treatment only considers the asymptotic small-distance behavior of one-particle Green's functions and does not deal with convergence issues of perturbative expansions. The main objective of the present work is the development of tools that can be used to obtain the leading asymptotic terms in the small-distance behavior for a large class of Feynman diagrams.   
They will then be employed to reveal possible connections between their perturbative order and their singular behavior, and may help to improve the performance of numerical approximation schemes for many-particle Green's functions. This will be the subject of a future publication.     

From a functional analytic point of view, Feynman diagrams have a rather complicated structure,  reflecting the highly multilinear character of the corresponding spectral resolutions. This requires some modifications of the standard singular calculus based on algebras of linear differential and pseudo-differential operators. In particular, the appearance of locally averaged partial traces of resolvents complicates the asymptotic analysis. The exact meaning of this notion will be discussed in Section \ref{singularbehavior}. Suffice it to say here that local averages of singular kernel functions and corresponding traces are considered in the theory of Hilbert-Schmidt operators, cf.~\cite{B88}, for further details. The novelty in our particular application is that the averages are restricted to subdomains of codimension greater than one.      
At the level of second-order perturbation theory, these locally averaged partial traces lead to a logarithmic term in the asymptotic expansion of, for example, one-particle Green's functions. We illustrate this particular aspect with some analytical
calculations for corresponding model problems.
 
In the final Section \ref{asymptoticsmoothness}
we put our results into a broader perspective and briefly outline some further developments.

\section{Many-particle models in electronic structure theory}
Ground states and excited states in atoms, molecules and solids have been successfully simulated by a variety of many-particle models based on {\em first principles} derived from quantum theory. Let us mention only the {(time-dependent) density functional theory} ((TD)DFT), density matrix functionals and many-particle perturbation theory based on Green's functions, such as the GW approximation.
There are several other, so-called {ab initio} approaches, cf.~\cite{HJO} for a comprehensive discussion, which are directly based on the many-particle wavefunction, such as {configuration interaction} (CI), {multi configuration self consistent field} (MCSCF), {multi reference configuration interaction} (MRCI), and {coupled cluster} (CC) theory \cite{FHS15,FFH17,FFH21}, to name just a few popular acronyms. 

Despite their suggestive name, {first principle} models often lack, at least from a mathematical point of view, a clean and transparent derivation from their common origin, the many-particle Schr\"odinger equation. 
Therefore, systematic improvements for many popular {first principle} models are difficult to obtain. One possible remedy is to exploit the interrelationships between these models, which allow the transfer of insights from one model to another. It is well known that the {ab initio} CC, (TD)DFT and density matrix models have many interrelationships based on a common background in Green's function many-particle theory, cf.~\cite{NS92,ORR02}, on which we will focus on in the following.
For the improvement of existing and the development of new {first principle} models, it turns out to be advantageous to have knowledge of the analytic properties
of many-particle Green's functions and related quantities, independent of specific discretization schemes.
A common feature of all models in electronic structure theory is the Coulomb potential, which represents the fundamental interaction between particles. Due to its singular nature, all quantities involved in electronic structure theory, such as wavefunctions, density matrices, Green's functions, and response functions, have a certain singular structure in common. 
It manifests in a specific asymptotic behavior of these quantities near the coalescence points of particles.

\subsection{Some basic reduced quantities of many-particle models} 
\label{basicquantities}
The common origin of many-particle models considered in electronic structure theory is the Schr\"o\-din\-ger equation,
which has as solutions the wavefunctions $\Psi(x_1,\ldots,x_N,t)$ of an $N$-particle system. Because of the 
high dimensionality of wavefunctions it is often advantageous to consider a reduced quantity that depends only on $n << N$ particles and describes the action of these particles in the mean-field generated by the other particles. The most prominent examples in many-particle theory are Green's functions, with one- and two-particle Green's functions, in terms of Dirac's bra-ket notation, given by
\begin{equation}
 G_1(x_1,t_1|x_2,t_2) = -i \langle \Psi_0 | T \{ \psi(x_1,t_1) \psi^\dagger(x_2,t_2) \} | \Psi_0 \rangle,
\label{G1}
\end{equation}
\begin{equation}
 G_2(x_1,t_1,x_2,t_2|x_3,t_3,x_4,t_4) = - \langle \Psi_0 | T \{\psi(x_1,t_1) \psi(x_2,t_2) \psi^\dagger(x_4,t_4) \psi^\dagger(x_3,t_3) \}  | \Psi_0 \rangle,
\label{G2}
\end{equation}
where $\Psi_0$ denotes the stationary ground state wavefunction of the $N$-particle system,\footnote{Here we assume that the underlying Hamiltonian is time-independent.} $\psi^\dagger,\psi$ are field creation and annihilation operators in the Heisenberg picture\footnote{In order separate coordinates refering to anihilation and creation operators as arguments of Green's functions, we use the notation $G(\cdot | \cdot )$ such that coordinates of anihilation operators appear to the left and coordinates of creation operators appear to the right of the bar, respectively.} and $T \{ \cdots \}$ denotes the time-ordering of the operators involved. Therefore, a one-particle Green's function can be interpreted for $t_1 >t_2$ as the probability amplitude of a process where a particle is added to the ground state in $x_2$ at time  $t_2$ and a particle is annihilated in  $x_1$ at time  $t_1$, or vice versa if $t_2 >t_1$. 
Similarly, a two-particle Green's function describes the effect on the ground state of a pair of particles or holes that are created and then annihilated.

Alternatively, the one-particle Green's function $G_1$ can be viewed as the fundamental solution of a Schr\"odinger type equation, i.e.,
\begin{equation}
 \bigl( i\partial_{t_1} -{\mathfrak h}_1 \bigr)G_1(x_1,t_1|x_2,t_2) -\int \Sigma(x_1,t_1,x_3,t_3) G_1(x_3,t_3|x_2,t_2) dx_3 dt_3
 = \delta(x_1 -x_2) \delta(t_1-t_2) ,
\label{SchrG}
\end{equation}
where all many-particle effects are represented by the self energy $\Sigma$ and ${\mathfrak h}_1$ is the one-particle Hamiltonian associated to the coordinates of the first particle.
This equation reduces in the non-interacting case, i.e., when the self energy $\Sigma$ is set to zero, to
\begin{equation}
 \bigl( i\partial_{t_1} -{\mathfrak h}_1 \bigr)G^{(0)}_1(x_1,t_1|x_2,t_2) = \delta(x_1 -x_2) \delta(t_1-t_2).
\label{SchrG0}
\end{equation}
It is often convenient to perform a partial Fourier transform with respect to the time variable, i.e.
\begin{equation}
 G_1(x_1,x_2;\omega) = \int_{-\infty}^{\infty} e^{it\omega} G_1(x_1,t_1+t,x_2,t_1) \, dt .
\label{G1omega}
\end{equation}
After such a partial Fourier transform, (\ref{SchrG}) and (\ref{SchrG0}) become
\begin{equation}
 \bigl( \omega -{\mathfrak h}_1 \bigr)G_1(x_1,x_2;\omega) -\int \Sigma(x_1,x_3;\omega) G_1(x_3,x_2;\omega) dx_3
 = \delta^3(x_1 -x_2) ,
\label{SchrGomega}
\end{equation}
\begin{equation}
 \bigl( \omega -{\mathfrak h}_1 \bigr)G^{(0)}_1(x_1,x_2;\omega)
 = \delta(x_1 -x_2) .
\label{SchrG0omega}
\end{equation}

Furthermore, Green's functions satisfy Dyson-type integral equations, e.g.
\begin{multline}
 G_1(x_1,t_1|x_2,t_2) = G_1^{(0)}(x_1,t_1|x_2,t_2) \\ + \int G_1^{(0)}(x_1,t_1|x_3,t_3) \Sigma (x_3,t_3,x_4,t_4) G_1(x_4,t_4|x_2,t_2) \, dx_3 dx_4 dt_3 dt_4 .
 \label{DysonG1}
\end{multline}
While an analytic expression can be given for $G_1^{(0)}$, no such expression exists for the self energy. The latter is defined either via many-body perturbation theory, or Dyson's equation itself serves as a definition, given the definition
(\ref{G1}) of the interacting Green's function. In this sense $G_1$ and $\Sigma$ should be considered as dual quantities related by a Dyson equation.

Green's functions are related to other reduced quantities such as densities
\begin{equation}
 \rho(x,t) = -i \lim_{\tau \rightarrow 0_+} G_1(x,t|x,t+\tau) ,
\label{G1rho}
\end{equation}
or reduced $n$-particle density matrices, such as 
\begin{equation}
 \gamma(x_1,x_2,t) = -i \lim_{\tau \rightarrow 0_+} G_1(x_1,t|x_2,t+\tau) ,
\label{G1gamma}
\end{equation}
where both derived quantities can themselves serve as basic reduced quantities in many-particle models, such as in the framework of density and density-matrix functional theory.

\subsection{Spectral resolutions and their translation back to operator algebras}
\label{spectralresolution}
Modern techniques in quantum many-particle theory rely heavily on the formalism of second quantization. This formalism is based on spectral resolutions with respect to an illusory complete single-particle basis and corresponding multilinear representations of the quantities of interest. First, it provides powerful techniques for formal manipulations that take into account the bosonic or fermionic character of the particles involved. Furthermore, the second quantization provides multilinear expressions for the quantities of interest, which can be used in numerical simulations. In fact, it is hard to imagine how to avoid the formalism of second quantization altogether, especially within the diagrammatic approach that is considered the backbone of many-particle theory.
But besides the lack of mathematical rigor with respect to the underlying spectral resolution, such an approach hides the asymptotic information we are interested in.
Therefore, it is necessary to translate second quantized expressions back to an operator formalism that is convenient for our purposes\footnote{See \cite{CGS16} for an alternative analytical approach to deal with the second quantized formalism.}. Specifically, we have to deal with sums over virtual orbitals and energy denominators of the form 
\[
 \frac{1}{\varepsilon_a+\varepsilon_b-\varepsilon_i-\varepsilon_j}, \quad \lim_{\eta \searrow 0} \frac{1}{\varepsilon_a-\varepsilon_i-\omega \pm i\eta} ,
\]
where we adopt the common convention that sums with indices $a,b,c,\ldots$ and $i,j,k,\ldots$ pass over unoccupied and occupied orbitals, respectively. 
Let us develop a vocabulary for our translation of expressions derived from second quantization back to the operator formalism. It is a tacit agreement in the quantum many-particle literature that formal discrete sums actually represent generalized eigenfunction expansions in an infinite-dimensional one-particle Hilbert space ${\cal H}_1$,\footnote{The canonical choice for ${\cal H}_1$ is the function space $L^2(\mathbb{R}^3) \otimes \sigma$, where $\sigma$ denotes the discrete space of spin degrees of freedom. For notational simplicity, we suppress the spin degrees in the following.} with the corresponding one-particle Hamiltonian $\mathfrak{h}$, which has a countable discrete spectrum $\varepsilon_1 < \varepsilon_2 < \ldots < 0$ and a continuous spectrum $[0,\infty)$, here assuming the absence of a singular part of the spectrum. To get a complete one-particle basis in ${\cal H}_1$, it is therefore necessary to complement the eigenfunctions $\phi_{\alpha}$, which belong to the discrete eigenvalues $\varepsilon_{\alpha}$, $\alpha=1,2,\ldots$, by generalized eigenfunctions $\phi_{\lambda}$, with continuous index $\lambda \in [0,\infty)$, which represent the continuous part of the spectrum. To simplify the notation, it is common practice in the physics literature to deliberately choose summation signs and integral symbols to refer to spectral decompositions with respect to such a basis, cf.~\cite{FM68} [footnote p.~18]. 
Furthermore, in the present work we assume a spectral gap $\Delta_{ov}>0$ between occupied and unoccupied one-particle states in the system.  
We introduce the canonical orthogonal projection operators $\Qocc$ and $\Qvirt$, which project to the subspace ${\cal H}_{\mbox{occ}}:= \mbox{span}\{\phi_i\}_{i=1,n}$ spanned by occupied orbitals and its orthogonal complement ${\cal H}_{\mbox{virt}} :={\cal H}_{\mbox{occ}}^{\perp}$ in ${\cal H}_1$, respectively. The one-particle Hamiltonian $\mathfrak{h}$ represents a semi-bounded self-adjoint operator defined on a dense subspace ${\cal D}_1 \subset {\cal H}_1$ and has a natural decomposition with respect to the orthogonal subspaces ${\cal H}_{\mbox{occ}}$, ${\cal H}_{\mbox{virt}}$ via
\[
 \mathfrak{h} = \Qocc \mathfrak{h} \Qocc +\Qvirt \mathfrak{h} \Qvirt .
\]

To llustrate our approach, consider a typical two-particle term of the form 
\[
 \kappa(x,\tilde{x}) := \sum_{a,b,i,j} \frac{\phi_{i}(x) \phi_{a}(x) \phi_{j}(\tilde{x}) \phi_{b}(\tilde{x})}{\varepsilon_a+\varepsilon_b-\varepsilon_i-\varepsilon_j} F_{abij}
\]
with 
\[
 F_{abij} := \iint \phi_{i}(y) \phi_{a}(y) \phi_{j}(\tilde{y}) \phi_{b}(\tilde{y}) f(y,\tilde{y}) dyd\tilde{y} 
\]
for some given function $f(y,\tilde{y})$.
In operator form, this expression can be written as
\[
 \kappa(x,\tilde{x}) := \sum_{i,j} \phi_{i}(x) \otimes \phi_{j}(\tilde{x}) S_{ij} \bigl( \phi_{i}(x) \otimes \phi_{j}(\tilde{x}) f(x,\tilde{x}) \bigr)  .
\]
with
\begin{equation}
 S_{ij} := \Qvirt \otimes \Qvirt \frac{1}{{\mathfrak h}_1 +{\mathfrak h}_2 -\varepsilon_i-\varepsilon_j}  \Qvirt \otimes \Qvirt ,
\label{Sij}
\end{equation}
which is a bounded operator in ${\cal L}({\cal H}_1 \otimes {\cal H}_1)$, due to the gap $\Delta_{ov}$ in the spectrum of ${\mathfrak h}$. 
Here, ${\mathfrak h}_1$ and ${\mathfrak h}_2$ denote the one-particle Hamiltonian with respect to the coordinates of the first and second particle, respectively.
Consider the following partial differential equation
\begin{equation}
 \bigl( {\mathfrak h}_1 +{\mathfrak h}_2 -\varepsilon_i-\varepsilon_j \bigr) \tau_{ij}(x,\tilde{x}) = \phi_{i}(x) \otimes \phi_{j}(\tilde{x}) f(x,\tilde{x}) ,
\label{tau1}
\end{equation}
for the unknown $\tau_{ij}$ 
with generic singular right hand side, e.g.
\[
 f(x,\tilde{x}) =\frac{1}{|x-\tilde{x}|} ,
\]
which is treated in the framework of singular analysis below.
Here we just want to mention that the partial differential operator on the left hand side of (\ref{tau1}) represents a semi-bounded self-adjoint operator defined on a dense subspace ${\cal D}_2 \subset {\cal H}_1 \otimes {\cal H}_1$ and the right hand side belongs to ${\cal H}_1 \otimes {\cal H}_1$.\footnote{This can be easily seen by explicitly taking ${\cal H}_1 := L^2(\mathbb{R}^3)$, which gives ${\cal H}_1 \otimes {\cal H}_1 =L^2(\mathbb{R}^3 \times \mathbb{R}^3)$.}

Multiplying (\ref{tau1}) from the left with $\Qvirt \otimes \Qvirt$ we get
\begin{equation}
 H_{ij} \tau_{ij}(x,\tilde{x}) = \Qvirt \otimes \Qvirt \bigl( \phi_{i}(x) \otimes \phi_{j}(\tilde{x}) f(x,\tilde{x}) \bigr) ,
\label{tau2}
\end{equation}
with
\[
 H_{ij} := \Qvirt \otimes \Qvirt \bigl( {\mathfrak h}_1 +{\mathfrak h}_2 -\varepsilon_i-\varepsilon_j \bigr) \Qvirt \otimes \Qvirt ,
\]
which satisfies the identity
\begin{equation}
 S_{ij} H_{ij} = \Qvirt \otimes \Qvirt ,
\label{SijHij}
\end{equation}
on ${\cal H}_1 \otimes {\cal H}_1$. Applying (\ref{SijHij}) to (\ref{tau2}), we get 
\[
 \Qvirt \otimes \Qvirt \tau_{ij}(x,\tilde{x}) = S_{ij} \bigl( \phi_{i}(x) \otimes \phi_{j}(\tilde{x}) f(x,\tilde{x}) \bigr) ,
\]
and finally
\begin{equation}
 \kappa(x,\tilde{x}) := \sum_{i,j} \phi_{i}(x) \otimes \phi_{j}(\tilde{x}) \Qvirt \otimes \Qvirt \tau_{ij}(x,\tilde{x}) .
\label{Sij->tauij}
\end{equation}
In summary, the central point of our approach is equation (\ref{tau1}), which allows us to extract the asymptotic small distance behavior of $\tau_{ij}$ and thus, via (\ref{Sij->tauij}), of $\kappa$, using techniques from singular analysis to be discussed in detail below.

\section{Singular analysis, the limiting absorption principle and pseudo-differential calculus}
\label{PsiDO}
In order to obtain the asymptotic small distance behavior of a given quantity, we have to identify the corresponding singular {partial differential equation} (PDE), like (\ref{tau1}), whose solution, after some intermediate steps, finally provides the desired asymptotic information. 
In general, such a PDE corresponds to a many-particle problem, which cannot be solved analytically. Instead of searching for an explicit solution, we are content with constructing a corresponding parametrix. Knowledge of a parametrix and its corresponding remainder, the so-called Green operator\footnote{The notion of a Green operator should not be confused with that of a Green's function. The latter is actually similar to a parametrix, cf.~\cite{FFH23} for further details.} turns out to be sufficent to obtain the singular parts of the asymptotic small-distance behavior of a solution of a singular PDE. The concept of a parametrix is non-standard in many-particle theory. Therefore, we want to briefly outline the underlying idea and an essential feature of it in singular analysis. A parametrix can be thought of as a pseudo-inverse\footnote{A related notion in linear algebra is the Moore-Penrose pseudo-inverse of a possibly non-invertible matrix.} of an elliptic partial differential operator $\Aop$. Ellipticity includes the Fredholm property, which means that the operator has a finite dimensional kernel and cokernel. Under such a premise, the existence of a left (right) parametrix $\Pop$ can be proved that satisfies the equation    
\[
 \Pop \Aop = \Iop +\Kop_l \quad \quad (\Aop \Pop = I +\Kop_r) ,
\]
respectively. The parametrix $\Pop$ can be represented as a pseudo-differential operator and provides an inverse modulo the compact operators $\Kop_l$, ($\Kop_r$). In the standard pseudo-differential calculus on smooth manifolds, these compact operators are smoothing operators that do not encode any specific asymptotic information. In contrast, the singular calculus applies to singular spaces with conical, edge, and corner singularities. In the singular case, the compact remainders are called Green operators $\Gop_l$, ($\Gop_r$), which now encode important asymptotic information. As an eplicit example, consider the action of a parametrix $\Pop$ on an equation of the type $\Aop u=f$ with possibly singular elliptic operator $\Aop$ and right hand side $f$. By left-multiplication with $\Pop$ we get the equation $\Pop \Aop u=\Pop f$,
and with $\Gop_l = \Pop \Aop -\Iop$, we finally get
\[
 u=\Pop f -\Gop_l u .
\]  
Roughly speaking, if $f$ has a well-defined asymptotic behavior, the parametrix $\Pop$ maps $f$ to another function whose asymptotic behavior is also well-defined. For the remainder $\Gop_l u$, we do not know the solution $u$, but regardless of its particular asymptotic behavior, the Green operator $\Gop_l$ maps it into a space with well-defined asymptotic behavior depending only on $\Aop$, we refer to the monographs \cite{ES97,HS08,Schulze98} for a detailed exposition. What remains to be done is an explicit construction of the parametrix and the corresponding Green operator, which can be achieved order by order with a recursive procedure for the asymptotic expansion of the parametrix, cf.~\cite{FHS16} for further details. 
 
\subsection{Weighted cone and edge Sobolev spaces with asymptotics} 
\label{weightedSobolev}
To extract the desired asymptotic information, it is necessary to introduce appropriate function spaces that take care of the asymptotic behavior near a singularity, and corresponding operator algebras that allow us to keep track of it. In the field of singular analysis one considers so-called {\em Sobolev spaces with asymptotics}, which are constructed in a recursive manner. Starting with point-like conic singularities, one  proceeds successively to higher order edge and corner singularities. In the present work, we are primarily concerned with edge-type singularities along the diagonals of the Green's functions under consideration.

Let us first consider conical singularities,  locally modeled by an open stretched cone $$X^\wedge :=\mathbb{R}_+ \times X$$ with base $X$.
To avoid unnecessary generality and complicated notation, we restrict ourselves to cones with base $X$ that are diffeomorphic to $S^2$. The 
weighted Sobolev spaces ${\cal K}^{s,\gamma}(X^\wedge)$ for $s \in \mathbb{N}_0$ are defined with respect to the corresponding spherical polar coordinates $\tilde{x} \rightarrow (r,x)$ as 
\[
 {\cal K}^{s,\gamma}(X^\wedge) := \omega {\cal H}^{s,\gamma}(X^\wedge) +(1-\omega) H^s(\mathbb{R}^{3}) ,
\]
for a cutoff function $\omega$, i.e.~$\omega\in C_0^\infty(\overline{\mathbb{R}}_+)$ such that $\omega(r)=1$ near $r=0$.
Here ${\cal H}^{s,\gamma}(X^\wedge) = r^\gamma {\cal H}^{s,0}(X^\wedge)$ and ${\cal H}^{s,0}(X^\wedge)$
is defined as the set of all $u(r,x) \in r^{-1} L^2(\mathbb{R}_+ \times X)$
such that $(r \partial_r)^jDu \in r^{-1} L^2(\mathbb{R}_+ \times X)$ for all $D \in \mbox{Diff}^{s-j}(X)$,
$0 \leq j \leq s$. The definition for $s \in \mathbb{R}$ follows by duality and complex interpolation.
Weighted Sobolev spaces with asymptotics are then subspaces of ${\cal K}^{s,\gamma}$ spaces defined as direct sums
\begin{equation}
 {\cal K}^{s,\gamma}_Q (X^\wedge) := {\cal E}^\gamma_Q (X^\wedge) + {\cal K}^{s,\gamma}_\Theta (X^\wedge) 
\label{E+K}
\end{equation}
of flattened weighted cone Sobolev spaces
\[
 {\cal K}^{s,\gamma}_\Theta (X^\wedge) := \bigcap_{\epsilon > 0} {\cal K}^{s,\gamma - \vartheta - \epsilon}
 (X^\wedge) 
\]
with $\Theta =(\vartheta,0]$, $-\infty \leq \vartheta < 0$, and
asymptotic spaces
\begin{equation}
 {\cal E}^\gamma_Q (X^\wedge) := \biggl\{ \omega(r) \sum_j \sum_{k=0}^{m_j} c_{jk}(x) r^{-q_j} \ln^k r \biggr\} .
\label{EQgamma}
\end{equation}
The asymptotic space ${\cal E}^\gamma_Q (X^\wedge)$ is characterized by a sequence $q_j \in \mathbb{C}$
which stems from a strip of the complex plane, i.e.
\[
 q_j \in \left\{ z: \frac{3}{2}-\gamma + \vartheta < \Re z < \frac{3}{2}-\gamma \right\} ,
\]
where the width and position of this strip are determined by its {\em weight data} $(\gamma,\Theta)$
with $\Theta =(\vartheta,0]$ and $-\infty \leq \vartheta < 0$. Each substrip of finite width
contains only a finite number of $q_j$. Furthermore, the coefficients
$c_{jk}$ belong to finite dimensional subspaces $L_j \subset C^\infty(X)$.
The asymptotics of ${\cal E}^\gamma_Q(X^\wedge)$ is thus completely
characterized by the {\em asymptotic type} $Q := \{(q_j,m_j,L_j)\}_{j \in \mathbb{Z}_+}$.
In the following we use the asymptotic subspaces
\[
 {\cal S}^\gamma_Q (X^\wedge) := \left\{ u \in {\cal K}^{\infty,\gamma}_Q (X^\wedge) :
 (1- \omega) u \in {\cal S}(\mathbb{R},C^\infty(X))|_{\mathbb{R}_+} \right\} 
\]
with Schwartz-type behavior for exit $r \rightarrow
\infty$. The spaces ${\cal K}^{s,\gamma}_Q(X^\wedge)$ and ${\cal
S}^\gamma_Q (X^\wedge)$ are Fr\'echet spaces equipped with
natural semi-norms according to the decomposition (\ref{E+K}); we
refer to \cite{ES97,HS08,Schulze98} for further details.

Weighted wedge Sobolev spaces on $\mathbb{W} := X^{\wedge} \times Y$ can then be defined as
functions $Y \rightarrow {\cal K}_{(Q)}^{s, \gamma}(X^\wedge)$, where the subscript $Q$ in parentheses means that the correpsonding expression refers to weighted Sobolev spaces with and without asymptotics as well. Here and in the following, upper-case subscripts like $Q$ denote the asymptotic type of cone spaces.
Consider the case $Y= \mathbb{R}^3$ and the corresponding wedge Sobolev spaces
\[
 {\cal W}^s(\mathbb{R}^3, {\cal K}_{(Q)}^{s, \gamma}(X^\wedge)) :=
 \{ u : \mathbb{R}^3 \rightarrow {\cal K}_{(Q)}^{s, \gamma}(X^\wedge)
 \, | \, u \in  \overline{{\cal S}(\mathbb{R}^3, {\cal K}_{(Q)}^{s, \gamma}(X^\wedge)} \} 
\]
with $s, \gamma \in \mathbb{R}$ and norm closure with respect to the norm
\[
 \| u \|_{{\cal W}^{s}(\mathbb{R}^3, {\cal K}_{(Q)}^{s, \gamma}(X^\wedge))}^2 := \int [\eta]^{2s} \| \kappa^{-1}_{[\eta]} (F_{y\rightarrow \eta} u)(\eta)
 \|_{{\cal K}^{s,\gamma}_{(Q)}(X^\wedge)}^2 d\eta .
\]
Here $F_{y\rightarrow \eta}$ denotes the Fourier transform in $\mathbb{R}^3$ and $\{ \kappa_\lambda \}_{\lambda \in \mathbb{R}_+}$
is a strongly continuous group of isomorphisms $\kappa_\lambda : {\cal K}_{(Q)}^{s, \gamma}(X^\wedge) \rightarrow 
{\cal K}_{(Q)}^{s, \gamma}(X^\wedge)$ defined for the three-diemnsional particle case by
\[
 \kappa_\lambda u(r,x,y) := \lambda^{\frac{3}{2}} u(\lambda r,x,y) .
\]
The function $[\eta]$ involved in the norm is given by a strictly positive $C^{\infty}(\mathbb{R}^{3})$ function of the covariables $\eta$ such that
$[\eta] = |\eta|$ for $|\eta| \geq \epsilon >0$. The motivation behind this group action is the twisted homogeneity
of the principal edge symbols, see~\cite{Schulze98} for more details. For  an open subset $Y \subset \mathbb{R}^3$ we define
\[
 {\cal W}^s_{\mbox{\footnotesize comp}}(Y, {\cal K}_{(Q)}^{s, \gamma}(X^\wedge)) :=
 \{ u \in {\cal W}^s(\mathbb{R}^3, {\cal K}_{(Q)}^{s, \gamma}(X^\wedge)): \supp u \subset Y \ \mbox{compact} \} ,
\]
and
\[
 {\cal W}^s_{\mbox{\footnotesize loc}}(Y, {\cal K}_{(Q)}^{s, \gamma}(X^\wedge)) :=
 \{ u \in {\cal D}'(Y, {\cal K}_{(Q)}^{s, \gamma}(X^\wedge)): \varphi u \in
 {\cal W}^s_{\mbox{\footnotesize comp}}(\mathbb{R}^3, {\cal K}_{(Q)}^{s, \gamma}(X^\wedge)) \ \mbox{for each} \ 
 \varphi \in C^\infty_0(Y) \} .
\]
The weighted Sobolev spaces ${\cal W}^{\infty}_{\mbox{\footnotesize comp}}(Y, {\cal K}_{Q}^{\infty, \gamma}(X^\wedge))$, which are of particular interest 
in our applications, have a nice tensor product representation for their asymptotic expansion, cf.~\cite{Schulze98} [Prop.~3.1.33], given by 
\begin{equation}
 \omega(r) \sum_{j} \sum_{k=0}^{m_j} r^{-p_j} \log^k r \, c_{jk}(x) v_{jk}(y) + h_{\Theta}(r,x,y)
\label{appasymp}
\end{equation}
where $(r,x,y)$ denotes the corresponding coordinates on the wedge $X^\wedge \times Y$. Tensor components $c_{jk} \in C^\infty(X)$, $v_{jk} \in H^\infty_{\mbox{\footnotesize comp}}(Y)$
correspond to functions on the base of the cone $X$ and the edge $Y$, respectively.

\subsection{The case of conical singularities}
\label{conicalsingularities}
To give a rough outline of our approach, for simplicity we first consider the free single-particle Hamiltonian $\mathfrak{h}_0 := -\frac{1}{2}\Delta$ and the corresponding resolvent
\[
 R_0(z) := \bigl( \mathfrak{h}_0 -z \bigr)^{-1}, \quad z \in \mathbb{C} \setminus \overline{\mathbb{R}}_+
\]
which is a bounded operator in $B(L^2(\mathbb{R}^3),L^2(\mathbb{R}^3))$ and $B(L^2(\mathbb{R}^3),H^2(\mathbb{R}^3))$, respectively. It has been shown in \cite{Agmon75} that even in the limit $z \rightarrow \lambda$, with $\lambda \in \mathbb{R}_+$, i.e.~in the continuous spectrum of $\mathfrak{h}_0$, the resolvent in uniform operator topology converges to a bounded operator between appropriately weighted function spaces. These functions spaces are defined for $s \in \mathbb{R}$ by
\begin{equation}
 L^{2,s}(\mathbb{R}^3) := \left\{ u \ | \ \bigl( 1+|x|^2 \bigr)^{\frac{s}{2}} u \in L^{2}(\mathbb{R}^3) \right\} ,
\label{L2s}
\end{equation}
\begin{equation}
 H^{2,s}(\mathbb{R}^3) := \left\{ u \ | \ D^{\alpha} u \in L^{2,s}(\mathbb{R}^3) \ \mbox{for} \ |\alpha| \leq 2 \right\} .
\label{H2s}
\end{equation}
According to \cite{Agmon75} [Theorem 4.1], the resolvent $R_0(z)$ can be considered as an analytic operator-valued function on $\mathbb{C} \setminus \overline{\mathbb{R}}_+$, with values in $B(L^{2,s},H^{2,-s})$ for $s>\frac{1}{2}$. For any $\lambda \in \mathbb{R}_+$, in the uniform operator topology of $B(L^{2,s},H^{2,-s})$,
the limits  \[
 R^{+}_0(\lambda) := \lim_{\substack{z \rightarrow \lambda \\ \Im z >0}} R_0(z), \quad
 R^{-}_0(\lambda) := \lim_{\substack{z \rightarrow \lambda \\ \Im z <0}} R_0(z) 
\]
exist.
Furthermore, for any $u \in L^{2,s}$, $s>\frac{1}{2}$, they satisfy the partial differential equation
\begin{equation}
 \bigl( \mathfrak{h}_0 -\lambda \bigr) R^{\pm}_0(\lambda) u =u .
\label{hRpm}
\end{equation}

Given $u \in S_P^{\gamma}(({\cal S}^2)^{\wedge})$ for some asymptotics $P$ with $\frac{1}{2} < \gamma < \frac{3}{2}$ and the two-sphere $S^2$ now being a concrete basis of the cone $X^{\wedge}$, we have $u \in L^{2,s}(\mathbb{R}^3)$ for $s \in \mathbb{R}$ and (\ref{hRpm}) satisfied. Now, what can be said about the asymptotic behavior of $R^{\pm}_0(\lambda) u \in H^{2,-s}$, $s>\frac{1}{2}$. Consider the decomposition of $R^{\pm}_0(\lambda)$ into a short-range and a long-range part 
\begin{equation}
 R^{\pm}_0(\lambda) u = \omega R^{\pm}_0(\lambda) u +(1-\omega) R^{\pm}_0(\lambda) u
\label{Rshortlong}
\end{equation}
for an arbitrary $C^{\infty}_0$ cutoff function $\omega$ that is equal to one on a sphere of radius $r_{\omega}$ around the origin. Putting this into (\ref{hRpm}) we get
\begin{equation}
 \tilde{\omega} \bigl( \mathfrak{h}_0 -\lambda \bigr) \omega R^{\pm}_0(\lambda) u =\tilde{\omega}u ,
\label{ohRpm}
\end{equation}
with $\tilde{\omega} \prec \omega$, i.e.~$\tilde{\omega} \omega =\tilde{\omega}$. 
Now consider the auxiliary operator $\Aop_0:={\mathfrak h}_0 +\mu$, with $\mu >0$, which is an elliptic element in the operator class
$C^2((S^2)^\wedge,\boldsymbol{g})$, $\boldsymbol{g} =(\gamma,
\gamma-2, \Theta)$ of the cone algebra for $\gamma\notin \mathbb{Z}+\frac{1}{2}$ and
$\Theta=(-\infty,0]$. Then $\Aop$ has a parametrix $\Pop_{\gamma}$ in the cone algebra\footnote{For notational simplicity, we suppress the $\mu$-dependence of the parametrix and the Green operator in our notation.}
that belongs to $C^{-2}((S^2)^\wedge,\boldsymbol{g})$,
$\boldsymbol{g} =(\gamma-2, \gamma, \Theta)$. It can be written in the general form
\begin{equation}
 \Pop_{\gamma} = \omega' r^2 \opm_M^{\gamma-3} \bigl( h^{(-1)}(r,w) \bigr) \tilde{\omega}' + \bigl( 1- \omega' \bigr)
 \tilde{\Pop} \bigl( 1- \hat{\omega}' \bigr) ,
\label{paraA}
\end{equation}
where $\omega'$, $\tilde{\omega}'$, $\hat{\omega}'$ are cutoff
functions satisfying $\omega' \prec \tilde{\omega}'$, $\hat{\omega}' \prec \omega'$, and $\tilde{\Pop}$ is a standard pseudo-differential operator of order $-2$ on $\mathbb{R}^3$. By
definition, the parametrix satisfies the equation
\[
 \Pop_{\gamma} \Aop_0 = 1 +\tilde{\Gop}_{\gamma} \ \ \mbox{with} \ \tilde{\Gop}_{\gamma} \in C_G((S^2)^\wedge, \boldsymbol{g}_l), \ \boldsymbol{g}_l=(\gamma,\gamma,\Theta) ,
\]
where the Green operator $\Gop_{\gamma}$ maps ${\cal
K}^{s,\gamma}((S^2)^\wedge)$ to ${\cal S}^\gamma_Q$ for some
discrete asymptotics $Q$. Applying (\ref{paraA}) to (\ref{ohRpm}) yields
\begin{eqnarray*}
 \Pop_{\gamma} \tilde{\omega}u & = &  \Pop_{\gamma} \tilde{\omega} \bigl( \mathfrak{h}_0 -\lambda \bigr) \omega R^{\pm}_0(\lambda) u \\
 & = & \Pop_{\gamma} \tilde{\omega} \bigl( \mathfrak{h}_0 +\mu \bigr) \omega R^{\pm}_0(\lambda) u -(\lambda+\mu) \Pop_{\gamma} \tilde{\omega} R^{\pm}_0(\lambda) u\\
 & = & \omega R^{\pm}_0(\lambda) u +\tilde{\Gop}_{\gamma} \omega R^{\pm}_0(\lambda) u -(\lambda+\mu) \Pop_{\gamma} \tilde{\omega} R^{\pm}_0(\lambda) u 
\end{eqnarray*}
and
\begin{equation}
 \omega R^{\pm}_0(\lambda) u = \Pop_{\gamma} \tilde{\omega}u -\tilde{\Gop}_{\gamma} \omega R^{\pm}_0(\lambda) u +(\lambda+\mu) \Pop_{\gamma} \tilde{\omega} R^{\pm}_0(\lambda) u ,
\label{Rpmu}
\end{equation}
where the first two terms on the right hand side have well-defined asymptotic behavior. Furthermore, we can conclude
\[
 R^{\pm}_0(\lambda) u \in H^{2,-s}(\mathbb{R}^3) \ \rightarrow \ \omega R^{\pm}_0(\lambda) u \in H^{2}(\mathbb{R}^3) \ \rightarrow \ \omega R^{\pm}_0(\lambda) u \in {\cal H}^{2,\gamma}\bigr( (S^2)^{\wedge} \bigr) ,
\]
for $s>\frac{1}{2}$ and weight\footnote{It follows from a Sobolev embedding theorem, see~\cite{GT98}, that $\omega R^{\pm}_0(\lambda) u \in L^{\infty}(\mathbb{R}^3)$ and therefore belongs to ${\cal H}^{2,\gamma}\bigl( (S^2)^{\wedge} \bigr)$ for $\gamma < \frac{3}{2}$, see  \cite{Schulze98} [Prop.~2.1.45]. }
$\frac{1}{2} < \gamma < \frac{3}{2}$.
Next, we shift the weight of the parametrix. This results in an additional Green operator, i.e.~we get
\[
 \Pop_{\gamma} = \Pop_{\gamma+2} +\Gop_{\gamma}
\]
with 
\begin{eqnarray*}
 \Gop_{\gamma} & := & \Pop_{\gamma} - \Pop_{\gamma+2} \\
 & = & \omega' \opm_M^{\gamma-1} \bigl( h^{(-1)}(w) \bigr) \tilde{\omega}'
 - \omega' \opm_M^{\gamma+1} \bigl( h^{(-1)}(w) \bigr) \tilde{\omega}' \\
 & = & \left[ \omega' \opm_M^{\gamma-1} \bigl( h^{(-1)}(w) \bigr) r^{2} \tilde{\omega}'
 - \omega' r^{2} \opm_M^{\gamma-1} \bigl( T^{-2} h^{(-1)}(w) \bigr)
 \tilde{\omega}' \right] r^{-2}.
\end{eqnarray*}
According to \cite[Proposition
2.3.69]{Schulze98}, the term in square brackets belongs to $C_G((S^2)^\wedge,\boldsymbol{g})$ with $\boldsymbol{g}=(\gamma,\gamma,(-\infty,0]).$
Applying this to the last term in (\ref{Rpmu}) yields
\begin{equation}
 \Pop_{\gamma} \omega R^{\pm}_0(\lambda) u = \Pop_{\gamma+2} \omega R^{\pm}_0(\lambda) u +\Gop_{\gamma} \omega R^{\pm}_0(\lambda) u ,
\label{gp2g}
\end{equation}
with $\Pop_{\gamma+2} \omega R^{\pm}_0(\lambda) u \in {\cal H}^{2,\gamma+2}\bigr( (S^2)^{\wedge} \bigr)$, which means that we have moved this term with uncontrolled asymptotic behavior from $\gamma$ to $\gamma+2$. Finally, inserting (\ref{gp2g}) into (\ref{Rpmu}) yields
\begin{equation}
 \omega R^{\pm}_0(\lambda) u = \underbrace{\Pop_{\gamma} \tilde{\omega}u -\tilde{\Gop}_{\gamma} \omega R^{\pm}_0(\lambda) u +(\lambda+\mu)\Gop_{\gamma} \omega R^{\pm}_0(\lambda) u}_{:=R_{\gamma}u} +(\lambda+\mu) \Pop_{\gamma+2} \tilde{\omega} R^{\pm}_0(\lambda) u ,
\label{Rpmu2}
\end{equation}
We can now apply a bootstrap argument using (\ref{Rpmu2}) recursively to shift the weight $\gamma$ to even larger values, i.e.
\begin{eqnarray*}
 \omega R^{\pm}_0(\lambda) u & = & R_{\gamma}u +(\lambda+\mu)\Pop_{\gamma+2} \tilde{\omega}\bigl[ R_{\gamma}u +(\lambda+\mu) \Pop_{\gamma+2} \tilde{\omega} R^{\pm}_0(\lambda) u \bigr] \\
 & = & R_{\gamma}u +(\lambda+\mu) \underbrace{\bigl[ \Pop_{\gamma+2} \tilde{\omega} R_{\gamma}u +(\lambda+\mu) \Gop_{\gamma+2} \tilde{\omega}\Pop_{\gamma+2} \tilde{\omega} R^{\pm}_0(\lambda) u \bigr]}_{:=R_{\gamma+2}u}  \\
 & & +(\lambda+\mu)^2 \Pop_{\gamma+4} \tilde{\omega} \Pop_{\gamma+2} \tilde{\omega} R^{\pm}_0(\lambda) u
\end{eqnarray*}
and continue with
\begin{eqnarray*}
 \omega R^{\pm}_0(\lambda) u & = & R_{\gamma}u +(\lambda+\mu)R_{\gamma+2}u +(\lambda+\mu)^2 \Pop_{\gamma+4} \tilde{\omega} \Pop_{\gamma+2} \tilde{\omega} \bigl[ R_{\gamma}u +(\lambda+\mu) \Pop_{\gamma+2} \tilde{\omega} R^{\pm}_0(\lambda) u \bigr] \\
 & = & R_{\gamma}u +(\lambda+\mu)R_{\gamma+2}u +(\lambda+\mu)^2 \bigl[ \Pop_{\gamma+4} \tilde{\omega} \Pop_{\gamma+2} \tilde{\omega} R_{\gamma}u \\
	&&+	(\lambda+\mu) \Pop_{\gamma+4} \tilde{\omega} \Gop_{\gamma+2} \tilde{\omega} \Pop_{\gamma+2} \tilde{\omega} R^{\pm}_0(\lambda) u 
   +(\lambda+\mu) \Gop_{\gamma+4} \tilde{\omega}\Pop_{\gamma+4} \tilde{\omega}\Pop_{\gamma+2} \tilde{\omega} R^{\pm}_0(\lambda) u \bigr]\\
	&&+(\lambda+\mu)^3 \Pop_{\gamma+6} \tilde{\omega} \Pop_{\gamma+4} \tilde{\omega} \Pop_{\gamma+2} \tilde{\omega} R^{\pm}_0(\lambda) u,
\end{eqnarray*}
where we have used
\begin{eqnarray*}
 \Pop_{\gamma+4} \tilde{\omega} \Pop_{\gamma+2} \tilde{\omega} \Pop_{\gamma+2} \tilde{\omega} R^{\pm}_0(\lambda) u & = & \Pop_{\gamma+4} \tilde{\omega} \bigl[ \Pop_{\gamma+4} +\Gop_{\gamma+2} \bigr]\tilde{\omega} \Pop_{\gamma+2} \tilde{\omega} R^{\pm}_0(\lambda) u \\
 & = & \Pop_{\gamma+4} \tilde{\omega} \Gop_{\gamma+2} \tilde{\omega} \Pop_{\gamma+2} \tilde{\omega} R^{\pm}_0(\lambda) u +\Gop_{\gamma+4} \tilde{\omega} \Pop_{\gamma+4} \tilde{\omega} \Pop_{\gamma+2} \tilde{\omega} R^{\pm}_0(\lambda) u \\
 & & +\Pop_{\gamma+6} \tilde{\omega} \Pop_{\gamma+4} \tilde{\omega} \Pop_{\gamma+2} \tilde{\omega} R^{\pm}_0(\lambda) u .
\end{eqnarray*}
For our applications, the previous considerations must be generalized by including a local potential, i.e.
\begin{equation}
 \mathfrak{h} := -\tfrac{1}{2}\Delta +v
\label{hv1}
\end{equation}
so that the spectrum of $\mathfrak{h}$ consists of a discrete part $\sigma_{disc} \subset \mathbb{R}_-$ with the lowest eigenvalue $\epsilon_1$ and a continuous part $\sigma_{cont} = \overline{\mathbb{R}}_+$. Moreover, the absence of a singular part will be assumed here, and we assume that the potential $v$ is smooth\footnote{For the reduced quantities considered below, we want to avoid introducing additional singularities beyond the leading singularity along their diagonals $x=\tilde{x}$.
A smooth pseudopotential or finite nucleus model would do the job.} 
and satisfies the requirements of \cite{Agmon75}, in particular $|v(x)| \lesssim \bigl| 1+|x| \bigr|^{-1-\epsilon}$, for $\epsilon >0$, but see \cite{AK92}. According to \cite{Agmon75} [Theorem 4.2], our previous discussion can be applied literally to $\mathfrak{h}$ and the corresponding analytic resolvent 
\[
 R(z) := \bigl( \mathfrak{h} -z \bigr)^{-1}, \quad z \in \mathbb{C} \setminus \overline{\mathbb{R}}_+ \cup \sigma_{disc} ,
\]
and limits $\lambda \in \mathbb{R}_+$
\[
 R^{+}(\lambda) := \lim_{\substack{z \rightarrow \lambda \\ \Im z >0}} R(z), \quad
 R^{-}(\lambda) := \lim_{\substack{z \rightarrow \lambda \\ \Im z <0}} R(z) .
\]
For the actual calculation of a parametrix, we refer to \cite{FHSS11,FFH23}, where explicit asymptotic parametrix constructions for $\mathfrak{h}$ and $\mathfrak{h}_0$  with Coulomb potential have been discussed. The smooth potentials considered in the present work can be treated analogously. Alternatively, it maybe preferable to use the formal recursion relation 
\[
 (\mathfrak{h} -z)^{-1} = (\mathfrak{h}_0 -z)^{-1} -(\mathfrak{h}_0 -z)^{-1} v(x) (\mathfrak{h} -z)^{-1}, \quad  \quad z \in \mathbb{C} \setminus \overline{\mathbb{R}}_+ \cap \sigma_{disc} ,  
\]
in a recursive way to derive the asymptotic behavior and reduce everything to the resolvents $R^{\pm}_0$.
Actually, this recursion relation can be applied to the resolvents $R^{\pm}_0$, $R^{\pm}$ as shown in \cite{Agmon75}, where the recursion relation 
\[
 R^{\pm}(\lambda) = R^{\pm}_0(\lambda) +R^{\pm}_0(\lambda) v R^{\pm}(\lambda) 
\] 
was proved to hold
for a sufficiently fast decaying potential $v$.

We can summarize the content of this section in the following lemma.

\begin{lemma}
\label{lemmaRcone}
Let $u \in L^{2,s}(\mathbb{R}^3)$ be a function with specific asymptotic behavior, i.e.~$u$ belongs to a space  ${\cal K}^{s,\gamma}_Q (X^\wedge)$ of asymptotic type $Q$, see~Section \ref{weightedSobolev}. The short-range part of the resolvents $\omega R^{\pm}_0(\lambda)$, $\omega R^{\pm}(\lambda)$, cf.~(\ref{Rshortlong}), map $u$ into
a ${\cal K}^{s,\gamma}_{\tilde{Q}} (X^\wedge)$ of asymptotic type $\tilde{Q}$ via\footnote{The subscrift $(0)$ indicates that the asymptotic expansion can be applied to $R^{\pm}_{0}(\lambda)$ and $R^{\pm}(\lambda)$ as well.} 
\begin{equation}
 \omega R^{\pm}_{(0)}(\lambda) u \sim \omega R_{\gamma}(\lambda,\mu)u +(\lambda+\mu) \omega R_{\gamma+2}(\lambda,\mu)u +(\lambda+\mu)^2 \omega R_{\gamma+4}(\lambda,\mu)u + \cdots ,
\label{Rgamma}
\end{equation}
with
\begin{eqnarray*}
 R_{\gamma}(\lambda,\mu) & := & \Pop_{\gamma} \tilde{\omega} -\tilde{\Gop}_{\gamma} \omega R^{\pm}_{(0)}(\lambda) +(\lambda+\mu)\Gop_{\gamma} \omega R^{\pm}_{(0)}(\lambda) , \\
 R_{\gamma+2}(\lambda,\mu) & := & \Pop_{\gamma+2} \tilde{\omega} R_{\gamma}(\lambda,\mu) +(\lambda+\mu) \Gop_{\gamma+2} \tilde{\omega}\Pop_{\gamma+2} \tilde{\omega} R^{\pm}_{(0)}(\lambda) , \\
 R_{\gamma+4}(\lambda,\mu) & := &  \Pop_{\gamma+4} \tilde{\omega} \Pop_{\gamma+2} \tilde{\omega} R_{\gamma}(\lambda,\mu) +(\lambda+\mu) \Pop_{\gamma+4} \tilde{\omega} \Gop_{\gamma+2} \tilde{\omega} \Pop_{\gamma+2} \tilde{\omega} R^{\pm}_{(0)}(\lambda) \\
 & & +(\lambda+\mu) \Gop_{\gamma+4} \tilde{\omega}\Pop_{\gamma+4} \tilde{\omega}\Pop_{\gamma+2} \tilde{\omega} R^{\pm}_{(0)}(\lambda) , \\
 & \vdots & 
\end{eqnarray*}
where $\Pop_{\gamma}, \Pop_{\gamma+2}, \Pop_{\gamma+4}, \ldots$ are parametrices and $\tilde{\Gop}_{\gamma}, \Gop_{\gamma}, \Gop_{\gamma+2}, \Gop_{\gamma+4}, \ldots$ corresponding Green operators of the partial differential operator ${\mathfrak h}_{0} +\mu$ or ${\mathfrak h} +\mu$, $\mu >0$, depending on which variant, i.e.~$R^{\pm}_0(\lambda)$ or $R^{\pm}(\lambda)$, of (\ref{Rgamma}) has been considered.
\end{lemma}

\begin{remark}
\label{remarkgreen}
For smooth potentials $v$ in (\ref{hv1}), the Green operators in the asymptotic expansion (\ref{Rgamma}) map to spaces ${\cal K}^{s,\gamma}_{\tilde{Q}} (X^\wedge)$ with smooth asymptotic type $\tilde{Q}$, which means that a function $f \in {\cal E}^{s,\gamma}_{\tilde{Q}} (X^\wedge)$, cf.~(\ref{EQgamma}), can be extended to a function in $C^{\infty}(\mathbb{R}^3)$ after changing to Cartesian coordinates. It should be noted, however, that the pseudo-differential calculus also allows for singular potentials, e.g.~of the  Coulomb-type, see~\cite{FHSS11,FSS08} for specific applications.
\end{remark}

\subsection{The case of edge-type singularities}
\label{edgesingularities}
A situation analogous to the conical case mentioned above arises for edge-type singularities, where one considers two-particle systems in $\mathbb{R}^3 \times \mathbb{R}^3$ with an edge along the diagonal. Perturbation theory needs analytic resolvents
\[
 R^{(2)}_{(0)}(z) := \bigl( \mathfrak{h}_{(0),1} +\mathfrak{h}_{(0),2} -z \bigr)^{-1}, \quad z \in \mathbb{C} \setminus \bigl( \overline{\mathbb{R}}_+ \cup \sigma_{disc} \bigr)
\]
of the (free) two-particle Hamiltonian $\mathfrak{h}_{(0),1} +\mathfrak{h}_{(0),2}$ with $${\mathfrak h}_{0,i} :=-\tfrac{1}{2} \Delta_i \quad \mbox{and} \quad  {\mathfrak h}_i := {\mathfrak h}_{0,i} +v_i,\ i=1,2,$$
which are bounded operators in $B\bigl( L^2(\mathbb{R}^3 \times \mathbb{R}^3),L^2(\mathbb{R}^3 \times \mathbb{R}^3) \bigr)$ and $B\bigl( L^2(\mathbb{R}^3 \times \mathbb{R}^3),H^2(\mathbb{R}^3 \times \mathbb{R}^3) \bigr)$.
According to the limiting absorption principle, for $R^{(2)}_{0}(z)$ and any $\lambda \in \mathbb{R}_+$,  the  limits
\begin{equation}
 R^{(2),+}_0(\lambda) := \lim_{\substack{z \rightarrow \lambda \\ \Im z >0}} R^{(2)}_0(z), \quad
 R^{(2),-}_0(\lambda) := \lim_{\substack{z \rightarrow \lambda \\ \Im z <0}} R^{(2)}_0(z) .
\label{R20limits}
\end{equation}
exist in the uniform operator topology of $B(L^{2,s},H^{2,-s})$.
For these limits, the discussion in Section \ref{conicalsingularities} can be adopted almost verbatim, replacing the conical Sobolev spaces by the wedge Sobolev spaces discussed in Section \ref{weightedSobolev}, and the conical pseudo-differential calculus by a generalized edge pseudo-differential calculus. Since we do not make explicit use of the edge calculus in the rest of this paper, we only refer to the monographs \cite{HS08,Schulze98} for a general exposition and \cite{FH10,FHS15,FHS16,FFH17,FFHS20} for applications in electronic structure theory.

In the case of $R^{(2)}(z)$, however, the limiting absorption principle is not directly applicable. This is due to the fact that the potential part $v(x_1)+v(x_2)$, $(x_1,x_2) \in \mathbb{R}^3 \times \mathbb{R}^3$ does not satisfy $|v(x_1)+v(x_2)| \lesssim \bigl| 1+|(x_1,x_2)| \bigr|^{-1-\epsilon}$, for some $\epsilon >0$, even if $|v(x)| \lesssim \bigl| 1+|x| \bigr|^{-1-\epsilon}$ is satisfied by the potential $v$ for some $\epsilon >0$. 
Since we are only interested in the asymptotic behavior along the diagonal, we use the formal expansion
\begin{equation}
 \bigl( {\mathfrak h}_1 +{\mathfrak h}_2 -z \bigr)^{-1} =\bigl( {\mathfrak h}_{0,1} +{\mathfrak h}_{0,2} -z \bigr)^{-1} -\bigl( {\mathfrak h}_{0,1} +{\mathfrak h}_{0,2} -z \bigr)^{-1} \bigl( v_1 +v_2 \bigr) \bigl( {\mathfrak h}_1 +{\mathfrak h}_2 -z \bigr)^{-1} ,
\label{hv2hexpansion}
\end{equation}
and introduce appropriate cutoff functions $\omega_{12} \in C^{\infty}_0(\mathbb{R}^3 \times \mathbb{R}^3)$, with $\omega_{12}(x_1,x_2) =1$ for $|x_1-x_2| <c_1$ and $\omega_{12}(x_1,x_2) =0$ for $|x_1-x_2| > c_2$, with $0 <c_1 <c_2$.
Now we consider the first order term in the expansion (\ref{hv2hexpansion}), i.e.
\begin{eqnarray}
\nonumber
 \lefteqn{\bigl( {\mathfrak h}_{0,1} +{\mathfrak h}_{0,2} -z \bigr)^{-1} \bigl( v_1 +v_2 \bigr) \bigl( {\mathfrak h}_{0,1} +{\mathfrak h}_{0,2} -z \bigr)^{-1} \omega_{12}} \hspace{3cm} \\ \label{tildeom}
 & = &   \bigl( {\mathfrak h}_{0,1} +{\mathfrak h}_{0,2} -z \bigr)^{-1} \bigl( v_1 +v_2 \bigr) \tilde{\omega}_{12} \bigl( {\mathfrak h}_{0,1} +{\mathfrak h}_{0,2} -z \bigr)^{-1} \omega_{12} \\ \nonumber
 & & +\bigl( {\mathfrak h}_{0,1} +{\mathfrak h}_{0,2} -z \bigr)^{-1} \bigl( v_1 +v_2 \bigr) \bigl( 1 -\tilde{\omega}_{12} \bigr) \bigl( {\mathfrak h}_{0,1} +{\mathfrak h}_{0,2} -z \bigr)^{-1} \omega_{12} ,
\end{eqnarray}
with $\omega_{12} \prec \tilde{\omega}_{12}$. The multiplicative operator $\bigl( v_1 +v_2 \bigr) \tilde{\omega}_{12}$ in the first term on the right hand side of (\ref{tildeom}) satisfies an estimate of the form 
\begin{equation}
 |\bigl( v(x_1) +v(x_2) \bigr) \tilde{\omega}_{12}(x_1,x_2)| \lesssim \bigl| 1+|(x_1,x_2)| \bigr|^{-1-\epsilon} ,
\label{vx12est}
\end{equation} 
for some $\epsilon >0$, if $|v(x)| \lesssim \bigl| 1+|x| \bigr|^{-1-\epsilon}$ with $x \in {\mathbb R}^3$ is satisfied for some $\epsilon >0$.
From (\ref{vx12est}) it follows that $\bigl( v_1 +v_2 \bigr) \tilde{\omega}_{12}$ belongs to $B(H^{2,-s}, L^{2,s})$ for some $s>\frac{1}{2}$, and by submultiplicativity of the operator norm we get
\[
 \lim_{\substack{z \rightarrow \lambda \\ \Im z >0}} \bigl( {\mathfrak h}_{0,1} +{\mathfrak h}_{0,2} -z \bigr)^{-1} \bigl( v_1 +v_2 \bigr) \tilde{\omega}_{12} \bigl( {\mathfrak h}_{0,1} +{\mathfrak h}_{0,2} -z \bigr)^{-1} \omega_{12} = R^{(2),+}_0(\lambda) \bigl( v_1 +v_2 \bigr) \tilde{\omega}_{12} R^{(2),+}_0(\lambda) \omega_{12} ,
\]
and the corresponding limit for $\Im z <0$ in the uniform operator topology.
The second term on the right hand side of (\ref{tildeom}) represents a smoothing operator. This can be seen by looking at the kernel function of the operator
\[
 \bigl( 1 -\tilde{\omega}_{12} \bigr) \bigl( {\mathfrak h}_{0,1} +{\mathfrak h}_{0,2} -z \bigr)^{-1} \omega_{12} ,
\]
which is given by 
\begin{equation}
 \bigl( 1 -\tilde{\omega}_{12}(x_1,x_2) \bigr) k(x_1,x_2|\tilde{x}_1,\tilde{x}_2,z) \omega_{12}(\tilde{x}_1,\tilde{x}_2)
\label{kxtx}
\end{equation}
where $k \in C^{\infty}(\mathbb{R}^6 \times \mathbb{R}^6 \setminus D)$ is  singular along the diagonal $D$. Because of $\omega_{12} \prec \tilde{\omega}_{12}$, the multiplicative operators on the left and right side of $k$ in (\ref{kxtx}) cut out a neighborhood of the diagonal $D$.  
Therefore, (\ref{kxtx}) and its limits $z \rightarrow \lambda$ belong to $C^{\infty}(\mathbb{R}^6 \times \mathbb{R}^6)$ and can be ignored when considering the singular asymptotic behavior modulo smooth contributions.

For edge-type singularities, we can summarize the previous considerations in the following proposition.

\begin{proposition}
\label{propositionRedge}
Let $u \in L^{2,s}(\mathbb{R}^3 \times \mathbb{R}^3)$ be a function with specific asymptotic behavior along the diagonal $D$, i.e., $u$ belongs to a wedge Sobolev space ${\cal W}^s(D, {\cal K}_{Q}^{s, \gamma}(X^\wedge))$ of asymptotic type $Q$, see Section~\ref{weightedSobolev}. 
The singular asymptotic behavior of $\tilde{\omega} R^{\pm}_0(\lambda) \omega u$ and $\tilde{\omega} R^{\pm}(\lambda) \omega u$, modulo smooth remainders, can be derived from the pseudo-differential edge calculus by constructions analogous to the case of a conical singularity. In the case of $\tilde{\omega} R^{\pm}(\lambda) \omega u$, we assume for the potential $v$ that the estimate $|v(x)| \lesssim \bigl| 1+|x| \bigr|^{-1-\epsilon}$ is satisfied for some $\epsilon >0$.
\end{proposition}

\section{Asymptotic behavior of one-particle Green's functions}
In this section we turn to the main topic of the present work, which is the small distance behavior of reduced quantities such as Green's functions in electronic structure theory.

The singular behavior of Feynman diagrams in quantum field theory can be classified according to Weinberg's theorem \cite{Weinb60}. It would be desirable to have a similar tool in electronic structure theory that provides a priori information about the regularity properties of these diagrams. A first step in this direction has been made in 
\cite{FFH17,FFH21}, where the asymptotic behavior and the Besov regularity of certain classes of Goldstone diagrams, i.e.~ring and ladder diagrams as well as combinations thereof, were derived. In the present work we extend these results to frequency-dependent Feynman diagrams to study the time-dependent models mentioned above in the long run. Finally, we want to use them to reduce the computational complexity of numerical algorithms for these models in the future.

As a first step in this direction, we apply the techniques from singular analysis, outlined in the previous Sections, to study the asymptotic small distance behavior of frequency-dependent one-particle Green's functions up to 2nd order perturbation theory.

\subsection{One-particle Green's function in 1st order perturbation theory}
Let us consider the one-particle Green's function in 1st order perturbation theory
\begin{equation}
 G_{1p}(x,\tilde{x},\omega) = -\sum_{j,a,b} \phi_a(x) \phi_b(\tilde{x}) \frac{\langle aj|V^{(2)}|bj \rangle -\langle aj|V^{(2)}|jb \rangle}{(\varepsilon_a-\omega -i\eta) (\varepsilon_b-\omega -i\eta)} 
\label{G1p}
\end{equation}
with the two-particle Coulomb potential $V^{(2)}$,
where the limit $\eta \searrow 0$ must be understood implicitly. 
The discrete representation (\ref{G1p}) can be translated back into the continuum by introducing the one-particle Hamiltonian (\ref{hv1}),
a multiplicative operator $\hat{V}_c$ corresponding to the Hartree potential
\[
 V_c(x) := \sum_j \int V^{(2)}(x,\tilde{x}) |\phi_j(\tilde{x})|^2 \, d\tilde{x} ,
\]
and the integral exchange operator $\hat{V}_x$ with the kernel function
\[
 V_x(x,\tilde{x}) := -\sum_j \phi_j(x) V^{(2)}(x,\tilde{x}) \phi_j(\tilde{x}) .
\]
Both operators depend on the two-particle Coulomb potential $V^{(2)}$ and
have already been discussed from a singular analysis point of view in \cite{FFH17,FFH21,FSS08}. For $\phi_j$ in the Schwartz space, the exchange operator $\hat{V}_x$ can be  viewed as a parameter-dependent kernel function of a classical pseudo-differential operator in the H\"ormander class $S^{-2}(\mathbb{R}^3 \times \mathbb{R}^3)$.
Next, we introduce the operator
\begin{equation}
 \hat{A}^{+}(\omega) := \Qvirt \bigl( {\mathfrak h} -\omega-i\eta \bigr)^{-1} \Qvirt ,
\label{Apomega}
\end{equation}
with $\eta >0$ and consider the operator 
\begin{equation}
 \hat{G}^{+}_{1p}(\omega) :=-\hat{A}^{+}(\omega) \bigl( \hat{V}_c +\hat{V}_x \bigr) \hat{A}^{+}(\omega) ,
\label{G+1p}
\end{equation}
where (\ref{G1p}) is  the kernel function in the limit $\eta \searrow 0$.
To apply the limiting absorption principle, i.e.-to take the limit $\eta \searrow 0$ in (\ref{G+1p}), we assume that the potential operator $\hat{V}_c +\hat{V}_x$ belongs to $B(H^{2,-s}, L^{2,s})$ for some $s>\frac{1}{2}$. Due to the exponential decay of the eigenfunctions $\phi_j$ this is obviously true for the exchange part $\hat{V}_x$ but not for the multiplicative part $\hat{V}_c$ because of the slow decay of the Coulomb potential. To satisfy this  requirement, we replace the long-range Coulomb potential in the following discussion by a screened Coulomb potential with sufficiently fast decay at infinity. Since we are only concerned with the effects of the singular short-range part of the Coulomb potential, such an approach seems to be reliable. According to the limiting absorption principle, the limit
\[
 \hat{A}_{0}(\omega) := \lim_{\eta \searrow 0} \hat{A}^{+}(\omega)
\]
exists in the uniform operator topology, with $\hat{A}_{0}(\omega) \in B(L^{2,s},H^{2,-s})$ for $s>\frac{1}{2}$.
The following lemma connects the one-particle Green's function (\ref{G1p}) to an operator.

\begin{lemma}
Let $\hat{V}_c +\hat{V}_x \in B(H^{2,-s}, L^{2,s})$ for some $s>\frac{1}{2}$. Then we get the limit 
	\[
 \hat{G}_{1p}(\omega) = lim_{\eta \searrow 0} \hat{G}^{+}_{1p}(\omega) 
\]
in the uniform operator topology,
with $\hat{G}_{1p}(\omega) \in B(L^{2,s},H^{2,-s})$ given by
\begin{equation}
 \hat{G}_{1p}(\omega) = -\hat{A}_{0}(\omega) \bigl( \hat{V}_c +\hat{V}_x \bigr) \hat{A}_{0}(\omega) .
\label{G1pop}
\end{equation}
The operator (\ref{G1pop}) has the one-particle Green's function (\ref{G1p}) as its kernel function.  
\label{lemmaG1p}
\end{lemma}
\begin{proof}
The proof is a simple consequence of the submultiplicativity of the corresponding operator norms, i.e.
\begin{eqnarray*}
 \| \hat{G}^{+}_{1p}(\omega) -\hat{G}_{1p}(\omega) \|_{L^{2,s}, H^{2,-s}} & = & \| \hat{A}^{+}(\omega) \bigl( \hat{V}_c +\hat{V}_x \bigr) \hat{A}^{+}(\omega) -\hat{A}_{0}(\omega) \bigl( \hat{V}_c +\hat{V}_x \bigr) \hat{A}_{0}(\omega) \|_ {L^{2,s}, H^{2,-s}} \\
 & \leq & \| \bigl( \hat{A}^{+}(\omega) - \hat{A}_{0}(\omega) \bigr) \bigl( \hat{V}_c +\hat{V}_x \bigr) \hat{A}^{+}(\omega) \|_ {L^{2,s}, H^{2,-s}} \\
  & & +\| \hat{A}_{0}(\omega) \bigl( \hat{V}_c +\hat{V}_x \bigr) \bigl( \hat{A}^{+}(\omega) -\hat{A}_{0}(\omega) \bigr) \|_ {L^{2,s}, H^{2,-s}} \\  
 & \leq & \bigl( \| \hat{A}^{+}(\omega) \|_{L^{2,s}, H^{2,-s}} + \| \hat{A}_{0}(\omega) \|_{L^{2,s}, H^{2,-s}} \bigr) \| \hat{V}_c +\hat{V}_x \|_{H^{2,-s},L^{2,s}} \\
 & & \times \| \hat{A}^{+}(\omega) -\hat{A}_{0}(\omega) \|_{L^{2,s}, H^{2,-s}}
\end{eqnarray*}
\end{proof}

Let us assume that $\hat{V}_c +\hat{V}_x$ is a bounded multiplicative operator between spaces  ${\cal K}^{s,\gamma}$ with equal weights but possibly different asymptotic types $Q$ and $\tilde Q$, respectively, i.e.
\[
 \hat{V}_c +\hat{V}_x: \quad {\cal K}^{s,\gamma}_Q (X^\wedge) \rightarrow {\cal K}^{s,\gamma}_{\tilde{Q}} (X^\wedge) ,
\]
which is the case, e.g. for a Yukawa-like screening of the Coulomb potential
\[
 V^{(2)}(x_1,x_2) = e^{-\alpha |x_1-x_2|}/|x_1-x_2| , \ \ \alpha >0 .
\] 
To apply Lemma \ref{lemmaRcone}, we introduce appropriate cutoff functions $\eta \in C^{\infty}_0(\mathbb{R}^3)$ with $\eta(x) =1$ for $|x| <c_1$ and $\eta(x) =0$ for $|x| > c_2$ with $0 <c_1 <c_2$, and consider in (\ref{G1pop}) the decomposition 
\begin{eqnarray*}
 -\hat{\eta} \hat{G}_{1p}(\omega) & = & \hat{\eta}\hat{A}_{0}(\omega) \tilde{\eta} \bigl( \hat{V}_c +\hat{V}_x \bigr) \eta \hat{A}_{0}(\omega) +\hat{\eta}\hat{A}_{0}(\omega) \underbrace{\tilde{\eta} \bigl( \hat{V}_c +\hat{V}_x \bigr) \bigl( 1-\eta \bigr)}_{\mbox{s.o.}} \hat{A}_{0}(\omega) \\
 & & +\underbrace{\hat{\eta}\hat{A}_{0}(\omega) \bigl( 1-\tilde{\eta} \bigr)}_{\mbox{s.o.}} \bigl( \hat{V}_c +\hat{V}_x \bigr) \eta \hat{A}_{0}(\omega)
 +\underbrace{\hat{\eta}\hat{A}_{0}(\omega) \bigl( 1-\tilde{\eta} \bigr)}_{\mbox{s.o.}} \bigl( \hat{V}_c +\hat{V}_x \bigr) \bigl( 1-\eta \bigr) \hat{A} _{0}(\omega) ,
\end{eqnarray*}
with $\hat{\eta} \prec \tilde{\eta} \prec \eta$. The bracketed expressions represent smoothing operators (s.o.), see the corresponding discussion in Section \ref{edgesingularities}, and the second, third, and fourth terms can be safely ignored if one considers the singular asymptotic behavior modulo smooth contributions. Thus we get the asymptotic expression
\[
 -\hat{\eta} \hat{G}_{1p}(\omega) \sim \hat{\eta}\hat{A}_{0}(\omega) \tilde{\eta} \bigl( \hat{V}_c +\hat{V}_x \bigr) \eta \hat{A}_{0}(\omega)  \quad \mbox{modulo s.o.} .
\]
Taking $R^{+}(\omega) :=\hat{A}_{0}(\omega)$, we can now apply Lemma \ref{lemmaRcone} and get
\[
 -\hat{\eta} \hat{G}_{1p}(\omega) \sim \hat{\eta} \bigl( R_{\gamma}(\omega,\mu) + \cdots \bigr) \tilde{\eta} \bigl( \hat{V}_c +\hat{V}_x \bigr) \eta \bigl( R_{\gamma}(\omega,\mu) + \cdots \bigr) + \cdots
\]
If we restrict ourselves to those terms that  
potentially contribute to the singular asymptotics, see~Remark \ref{remarkgreen}, we get 
\begin{equation}
 -\hat{\eta} \hat{G}_{1p}(\omega) \sim \hat{\eta} \Pop_{\gamma} \tilde{\eta} \hat{V}_c \eta \Pop_{\gamma} +\hat{\eta} \Pop_{\gamma}  \tilde{\eta} \hat{V}_x \eta \Pop_{\gamma}  + \cdots 
\label{G1psim}
\end{equation}
in leading order of the asymptotic expansion.
Alternatively, the parameterices in the asymptotic expansion can be considered as classical pseudo-differential operators, where the symbol of the first term belongs to the H\"ormander class $S^{-4}(\mathbb{R}^3 \times \mathbb{R}^3)$ and the symbol of the second term belongs to the H\"ormander class $S^{-6}(\mathbb{R}^3 \times \mathbb{R}^3)$. The classification scheme for symbols translates to the asymptotic smoothness classification for the corresponding kernel functions, discussed in Section \ref{asymptoticsmoothness}.

\subsection{2nd order terms of the one-particle Green's function}
\label{2ndGreen}
The approximate one-particle Green's function in 2nd order perturbation theory is of particular importance for applications in quantum chemistry, see e.g., \cite{NSC84,OT16,PPOV17,TPPV19} and the references cited therein.  
From a mathematical point of view, an interesting new feature appears in 2nd order Feynman diagrams where partial traces 
of resolvents complicate the asymptotic analysis. Apart from that, we can apply techniques from singular analysis, similar to those used in 1st order perturbation theory.

Let us consider the 2nd order $2p1h$ contribution of a particle system. It is given in physics notation by
\begin{equation}
 G_{2p1h}(x,\tilde{x},\omega) = -\sum_{j,a,b,c,d} \phi_c(x) \phi_d(\tilde{x}) \frac{\langle cj|V^{(2)}|ab \rangle \langle ab|V^{(2)}|dj \rangle -\langle cj|V^{(2)}|ab \rangle \langle ab|V^{(2)}|jd \rangle}{(\varepsilon_c-\omega -i\eta) (\varepsilon_a+\varepsilon_b-\varepsilon_j-\omega -i\eta) (\varepsilon_d-\omega -i\eta)} ,
\label{G2p1h}
\end{equation}
see \cite{NO88} [Eq.~3.58a].
To give it a rigorous formulation, in addition to (\ref{Apomega}) and the multiplication operator $\hat{V}^{(2)}$ that denotes multiplication with a singular two-particle Coulomb-type potential, let us introduce the operator
\begin{equation}
 \hat{S}_{j}^{+}(\omega) := \Qvirt \otimes \Qvirt \frac{1}{{\mathfrak h}_1 +{\mathfrak h}_2 -\varepsilon_j-\omega-i\eta} \Qvirt \otimes \Qvirt .
\label{Sijomega}
\end{equation}
The two-particle operator we want to consider in the following is given by
\begin{equation}
 \hat{H}_j^{+} (\omega) :=\hat{V}^{(2)} \, \hat{S}_{j}^{+}(\omega) \, \hat{V}^{(2)} , 
\label{Hj}
\end{equation}
whose kernel function is given by
\[
 H_j^{+}(x_1,x_2|x_3,x_4;\omega) =\sum_{a,b} \frac{\phi_a(x_1)\phi_b(x_2)V^{(2)}(x_1,x_2)V^{(2)}(x_3,x_4)\phi_a(x_3)\phi_b(x_4)}{\varepsilon_a+\varepsilon_b-\varepsilon_j-\omega -i\eta} .
\]
For an asymptotic analysis of (\ref{G2p1h}), we need to consider the partially contracted kernel functions
\begin{equation}
 H^{+}_d(x_1,x_3;\omega) :=\sum_j \int \phi_j(x_2) H_j^{+}(x_1,x_2|x_3,x_4;\omega) \phi_j(x_4) \, dx_2 dx_4 ,
\label{sumjHjd}
\end{equation}
\begin{equation}
 H^{+}_x(x_1,x_3;\omega) :=\sum_j \int \phi_j(x_2) H_j^{+}(x_1,x_2|x_4,x_3;\omega) \phi_j(x_4) \, dx_2 dx_4 ,
\label{sumjHje}
\end{equation}
representing the direct and exchange parts in (\ref{G2p1h}), respectively. The kernel functions (\ref{sumjHjd}) and (\ref{sumjHje}) represent up to a constant the second order $2p1h$ part of the self-energy, see \cite{NO88} [Eq.~3.58b].
Let us introduce the operators $\hat{H}^{+}_d(\omega)$, $\hat{H}^{+}_x(\omega)$, which are defined by their kernel functions (\ref{sumjHjd}) and (\ref{sumjHje}), respectively. Now we can define the operator
\[
 \hat{G}_{2p1h}(\omega) =-\hat{A}^{+}(\omega) \bigl( \hat{H}^{+}_d(\omega) -\hat{H}^{+}_x(\omega) \bigr) \hat{A}^{+}(\omega) ,
\]
which has the one-particle Green's function (\ref{G2p1h}) as its kernel function. 

In the following we consider (\ref{Hj}) and its kernel function from the point of view of singular analysis and give a rigorous asymptotic expression for (\ref{sumjHjd}) and (\ref{sumjHje}). In particular, we show that it can be considered, modulo smooth terms, as a parameter-dependent kernel function of a classical pseudo-differential operator in the H\"ormander class $S^{-3}(\mathbb{R}^3 \times \mathbb{R}^3)$.

\subsection{Singular behavior of partially contracted kernel functions}
\label{singularbehavior}
Understanding the singular behavior of the  partially contracted kernel functions (\ref{sumjHjd}), (\ref{sumjHje}) requires careful analysis that takes into account the various overlapping singular strata involved.
To get started, it is useful to consider a simple example where the partial contraction can be calculated analytically. 

\begin{example}
\label{example2} 
As an analytical example of a partially contracted kernel function, consider the fundamental solution of the $3d$ Laplace operator, given by
\begin{equation}
 H(x,\tilde{x}) =-\frac{1}{4\pi |x-\tilde{x}|} \quad \mbox{with} \ x:=(x_1,x_2,x_3), \ \tilde{x}:=(\tilde{x}_1,\tilde{x}_2,\tilde{x}_3) ,
\label{fsL3d}
\end{equation}
and perform the contraction in $1d$ with respect to the Gaussion function $\phi(x_3) =e^{-x_3^2}$, $\phi(\tilde{x}_3) =e^{-\tilde{x}_3^2}$, i.e.~we consider the kernel function 
\[
 K(x_1,x_2|\tilde{x}_1,\tilde{x}_2) := -\frac{1}{4\pi} \int_{-\infty}^{\infty} \frac{\phi(x_3) \phi(\tilde{x}_3)}{|x-\tilde{x}|} dx_3d\tilde{x}_3 ,
\]
which can be calculated analytically. For this we use the well-known integral representation
\begin{equation}
 \frac{1}{|x-\tilde{x}|} =\frac{1}{\sqrt{\pi}} \int_{-\infty}^{\infty} e^{-|x-\tilde{x}|^2t^2} dt ,
\label{CG}
\end{equation}
and determine the integral
\begin{eqnarray}
\nonumber
 I(t) & = & \frac{1}{\sqrt{\pi}} \int_{-\infty}^{\infty} e^{-x_3^2}  \left[ \int_{-\infty}^{\infty} e^{-(x_3-\tilde{x}_3)^2t^2} e^{-\tilde{x}_3^2} d\tilde{x}_3 \right] dx_3 \\ \nonumber
 & = & \frac{1}{\sqrt{\pi}} \int_{-\infty}^{\infty} e^{-x_3^2}  \left[ \sqrt{\frac{\pi}{1+t^2}} \, e^{-\frac{t^2}{1+t^2}x_3^2} \right] dx_3 \\ \label{It}
 & = & \sqrt{\frac{\pi}{1+2t^2}} .
\end{eqnarray}
With this, we can calculate the partially contracted kernel function
\begin{eqnarray}
\nonumber
 K(x_1,x_2|\tilde{x}_1,\tilde{x}_2) & = & -\frac{1}{4\pi} \int_{-\infty}^{\infty} \sqrt{\frac{\pi}{1+2t^2}} e^{-|x_{12}|^2 t^2} dt, \ \mbox{with} \ x_{12} :=(x_1-\tilde{x}_1,x_2-\tilde{x}_2) \\ \label{K2d}
 & = & -\frac{1}{4\sqrt{2\pi}} e^{\frac{|x_{12}|^2}{4}} K_0 \left( \frac{|x_{12}|^2}{4} \right) ,
\end{eqnarray}
where we used \cite{GR07} [3.462 (25)]  in the second step. Using the  asymptotic expansion \cite{Agmon75} [9.6.8] for the modified Bessel function $K_0$, we get
\[
 K(x_1,x_2|\tilde{x}_1,\tilde{x}_2) \sim \frac{1}{2\sqrt{2\pi}} \ln |x_{12}| .
\]
Such singular asymptotic behavior is quite gratifying because it is proportional to the Green function of the $2d$ Laplace operator. \footnote{This example shows that it is essential to contract with smeared-out functions $\phi$, because contraction with a delta distribution $\phi(x)=\delta (x-a)$, $a \in \mathbb{R}$, results in a kernel function 
$ K(x_1,x_2|\tilde{x}_1,\tilde{x}_2) = -\frac{1}{4\pi |x_{12}|}$.
Thus, this allows to treat a singular $2d$ kernel function in our calculus.} 
\end{example}

However, in our envisioned applications, such direct analytic calculations of partial contractions of kernel functions are not feasible. Instead we rely on techniques from singular analysis, briefly discussed in Section \ref{PsiDO}, which have already been used to analyze the asymptotic behavior of pair amplitudes \cite{FHS15,FFH17,FFH21} in coupled cluster theory. Such an approach  requires as an intermediate step a full contraction with an appropriate tensor product. To illustrate our approach, we reconsider the previous example. 

\begin{example}
\label{example3}
Given the kernel function (\ref{fsL3d}) from Example \ref{example2} and the two Gaussian functions
\[
 \delta_{a,\beta}(x) := \frac{\beta}{\pi} e^{-\beta|x-a|^2}, \quad \quad \delta_{b,\beta}(x) := \frac{\beta}{\pi} e^{-\beta|x-b|^2}, \quad x:=(x_1,x_2) \in \mathbb{R}^2
\]
which can in the limit $\beta \rightarrow \infty$ be considered as approximations to Dirac $delta$ distributions with support at $a:=(a_1,a_2)$, $b:=(b_1,b_2)$, respectively. 
The partially contracted kernel function (\ref{K2d}) can be recovered in the limit $\beta \rightarrow \infty$ from the $3d$-integral 
\[
 K(a_1,a_2|b_1,b_2) = \lim_{\beta \rightarrow \infty} -\frac{1}{4\pi} \int_{\mathbb{R}^3} \frac{\delta_{a,\beta}(x_1,x_2) \otimes \phi(x_3) \cdot \delta_{a,\beta}(\tilde{x}_1,\tilde{x}_2) \otimes \phi(\tilde{x}_3)}{|x-\tilde{x}|} dxd\tilde{x} .
\]
This is shown in Appendix \ref{partialcontraction} via an explicit calculation.
\end{example}  

Now consider the operators (\ref{Sijomega}) and (\ref{Hj}), where (\ref{Sijomega}) is sandwiched between two singular multiplicative operators $\hat{V}^{(2)}$. In order not to get lost in technicalities, we make some simplifications concerning these operators. Given
\[
 \Qvirt \otimes \Qvirt = I \otimes I -I \otimes \Qocc -\Qocc \otimes I +\Qocc \otimes \Qocc ,
\] 
we get for (\ref{Sijomega}) the decomposition
\begin{multline*}
 \hat{S}_{j}^{+}(\omega) := \bigl( {\mathfrak h}_1 +{\mathfrak h}_2 -\varepsilon_j-\omega-i\eta \bigr)^{-1}
  -I \otimes \Qocc \bigl( {\mathfrak h}_1 +{\mathfrak h}_2 -\varepsilon_j-\omega-i\eta \bigr)^{-1} I \otimes \Qocc \\
  \quad \quad \quad \quad -\Qocc \otimes I \bigl( {\mathfrak h}_1 +{\mathfrak h}_2 -\varepsilon_j-\omega-i\eta \bigr)^{-1} \Qocc \otimes I
  +\Qocc \otimes \Qocc \bigl( {\mathfrak h}_1 +{\mathfrak h}_2 -\varepsilon_j-\omega-i\eta \bigr)^{-1} \Qocc \otimes \Qocc ,
\end{multline*}
where the second and third terms correspond to effective single-particle operators and the fourth term is of finite rank. In the following, we restrict our discussion to the first term of this decomposition and 
and apply the expansion (\ref{hv2hexpansion}) to it.
Since we have assumed smooth single-particle potentials $v$ from the beginning, it is clear that the expansion generates a series of operators of increasing smoothness. Therefore, we restrict ourselves to the first term of the expansion. Furthermore, we assume that the multiplicative operator $\hat{V}^{(2)}$,
depends only on the interparticle distance, i.e.~it represents a bare or screened Coulomb potential with the asymptotic expansion
\[
 V^{(2)}(x_1,x_2) \sim \frac{1}{|x_1-x_2|} + {\cal O}(1).
\]
To summarize our simplifications, instead of (\ref{Hj}) we consider the operator 
\begin{equation}
 \hat{H}^{(0)}(\omega) :=\hat{V}^{(2)} \bigl( {\mathfrak h}_{0,1} +{\mathfrak h}_{0,2} -\varepsilon-\omega-i\eta \bigr)^{-1} \hat{V}^{(2)} .
\label{H0j}
\end{equation}
Motivated by Example \ref{example3}, we define the $3d$ test functions
\[
 \delta_{a,\beta}(x) := \left( \frac{\beta}{\pi} \right)^{\frac{3}{2}} e^{-\beta|x-a|^2}, \quad \quad \delta_{b,\beta}(x) := \left( \frac{\beta}{\pi} \right)^{\frac{3}{2}} e^{-\beta|x-b|^2}, \quad x \in \mathbb{R}^3
\]
and the integral corresponding to the direct part
\begin{equation}
  K_{\beta}(a,b) := \lim_{\beta \rightarrow \infty} \iint_{\mathbb{R}^3 \times \mathbb{R}^3} \delta_{a,\beta}(x_1) \phi(x_2) \bigl( \hat{H}^{(0)}(\omega) \delta_{b,\beta} \otimes \phi \bigr) (x_1,x_2) \, dx_1dx_2 .
  \label{Kbeta}
\end{equation}
and consider the asymptotic behavior of the expression
\begin{equation}
  K(a,b) := \lim_{\beta \rightarrow \infty} K_{\beta}(a,b).
\label{Kab}
\end{equation}
The asymptotics of $K$ can be obtained in several steps. In the first step, we consider
\begin{equation}
 \Psi_{b,\beta}(x_1,x_2) := \bigl( {\mathfrak h}_{0,1} +{\mathfrak h}_{0,2} -\varepsilon-\omega-i\eta \bigr)^{-1} \hat{V}^{(2)} \delta_{b,\beta}(x_1) \otimes \phi(x_2) ,
\label{h1h2int}
\end{equation}
for which we can determine the asymptotic behavior by considering the corresponding PDE
\begin{equation}
 \bigl( {\mathfrak h}_{0,1} +{\mathfrak h}_{0,1} -\varepsilon-\omega-i\eta \bigr) \Psi_{b,\beta}(x_1,x_2) = \frac{\delta_{b,\beta}(x_1) \phi(x_2)}{|x_1-x_2|} ,
\label{h1h2pde}
\end{equation}
in the context of singular analysis. To do this, we need to introduce
appropriate coordinates that reflect the singular structure. These are center of mass coordinates given by
\[
 z_1(x_1,x_2) = \tfrac{1}{\sqrt{2}} (x_1 - x_2), \quad z_2(x_1,x_2) = \tfrac{1}{\sqrt{2}} (x_1 + x_2) .
\]
Within the pseudo-differential algebra outlined in Section \ref{PsiDO}, (\ref{h1h2pde}) can be solved by an asymptotic parametrix construction and corresponding Green operators, cf.~Lemma \ref{lemmaRcone} and Proposition \ref{propositionRedge}, with leading order term\footnote{For notational simplicity we suppress the $\varepsilon-$ and $\omega$-dependence of the quantities involved.} given by
\[
 \Psi_{b,\beta} \sim \Pop \bigl( V^{(2)} \delta_{b,\beta} \otimes \phi \bigr) -\Gop \Psi_{b,\beta} ,
\]
where $\Psi_{b,\beta}$ corresponds to $R^{(2),+}_0(\varepsilon+\omega)\bigl( V^{(2)} \delta_{b,\beta} \otimes \phi \bigr)$ in the notation of Section \ref{edgesingularities}.
An explicit calculation of the asymptotic parametrix and the Green operator, see \cite{FFH17},
shows that the  asymptotic expansion\footnote{The first term in the asymptotic expansion vanishes for an electron pair in a triplet state. The asymptotic expansion takes into account only terms with zero relative angular momentum, as indicated by the projection operator $P_0$. For further details see \cite{FFH17}.}
\begin{equation}
 P_0 \tilde{\Psi}_{b,\beta}(z_1,z_2) \sim \tilde{\Psi}_{b,\beta}(0,z_2) +\tfrac{1}{\sqrt{2}} |z_1| \phi \bigl( \tfrac{1}{\sqrt{2}} z_2 \bigr) \delta_{b,\beta} \bigl( \tfrac{1}{\sqrt{2}} z_2 \bigr) +{\cal O}(|z_1|^2) 
\label{P0tildePsi}
\end{equation}
exists, where we define $\tilde{\Psi}_{b,\beta} \bigl( z_1(x_1,x_2),z_2(x_1,x_2) \bigr) := \Psi_{b,\beta}(x_1,x_2)$.
This gives us the asymptotic expansion\begin{multline}
 K_{\beta}(a,b) \sim  \iint_{\mathbb{R}^3 \times \mathbb{R}^3} \frac{\delta_{a,\beta}(x_1) \phi(x_2)}{|x_1-x_2|} \biggl[ \Psi_{b,\beta} \bigl( \tfrac{1}{2} (x_1 + x_2),\tfrac{1}{2} (x_1 + x_2) \bigr) \\
 +\tfrac{1}{2} |x_1-x_2| \phi \bigl( \tfrac{1}{2} (x_1+x_2) \bigr) \delta_{b,\beta} \bigl( \tfrac{1}{2} (x_1+x_2) \bigr) +{\cal O}(|z_1|^2) \biggr] \, dx_1dx_2
\label{Kbeta-asymp}
\end{multline}
for (\ref{Kbeta}) and we treat the contributions of the terms in the asymptotic expansion (\ref{P0tildePsi}) separately.
First, we deal with the parameter-dependent integral
\begin{eqnarray}
\nonumber
 I_1(a,b) & := & \lim_{\beta \rightarrow \infty} \iint_{\mathbb{R}^3 \times \mathbb{R}^3} \frac{\delta_{a,\beta}(x_1) \phi(x_2)}{|x_1-x_2|} \Psi_{b,\beta} \bigl( \tfrac{1}{2} (x_1 + x_2),\tfrac{1}{2} (x_1 + x_2) \bigr) \, dx_1dx_2 \\ \label{I1ab}
 & = & \lim_{\beta \rightarrow \infty} 4 \iint_{\mathbb{R}^3 \times \mathbb{R}^3} \frac{\delta_{a,\beta}(\tilde{x}_1) \phi(2\tilde{x}_2-\tilde{x}_1)}{|\tilde{x}_1-\tilde{x}_2|} \Psi_{b,\beta}(\tilde{x}_2,\tilde{x}_2) \, d\tilde{x}_1d\tilde{x}_2 .
\end{eqnarray}
Its asymptotic behavior for $|a-b| \rightarrow 0$ depends on the properties of the family of functions $\Psi_{b,\beta}$ along the diagonal, for $\beta$ sufficiently large. To illustrate our problem, we focus on the leading order term of the singular asymptotic expansion of this family. We write (\ref{h1h2int}) as an integral equation, i.e.
\begin{equation}
 \Psi_{b,\beta}(x_1,x_2) := \iint_{\mathbb{R}^3 \times \mathbb{R}^3} K_{12}(x_1,x_2\, | \, \hat{x}_1,\hat{x}_2) \frac{\delta_{b,\beta}(\hat{x}_1) \phi(\hat{x}_2)}{|\hat{x}_1-\hat{x}_2|} \, d\hat{x}_1d\hat{x}_2 ,
\label{h1h2kernel}
\end{equation}
where the kernel function $K_{12}$ of the shifted inverse Laplace operator
\[
 \bigl( {\mathfrak h}_{0,1} +{\mathfrak h}_{0,1} -\varepsilon-\omega-i\eta \bigr)^{-1} ,
\]
has an asymptotic behavior of the form\footnote{The kernel function $K_{12}$ can be obtained by analytic continuation of a fundamental solution of the shifted Lapalce operator in $6d$, i.e., $\Delta_{6}-\kappa^2$, given by 
\[
 -(2\pi)^{-3} \kappa^{2} r^{-2} K_{2}(\kappa r) ,
\]
where $K_{2}$ denotes a modified Bessel function of the second kind,
see e.g.~\cite{Schwartz} [Section II, \S 3]
The modified Bessel function has an asymptotic behavior for $r \rightarrow 0$, cf.~\cite{AS}, of the form
$K_{2}(\kappa r) \sim 2 \left( \kappa r \right)^{-2}$. 
}
\begin{equation}
 K_{12}(x_1,x_2\, | \, \hat{x}_1,\hat{x}_2) \sim \frac{c_{12}}{\bigl( |x_1-\hat{x}_1|^2 +|x_2-\hat{x}_2|^2 \bigr)^2} .
\label{K12asymp}
\end{equation}
We thus restrict our discussion to the leading order asymptotic expression
\begin{eqnarray}
\nonumber
 \Psi^{(0)}_{b,\beta}(\tilde{x}_2,\tilde{x}_2) & \sim & \iint_{\mathbb{R}^3 \times \mathbb{R}^3} \frac{c_{12}}{\bigl( |\tilde{x}_2-\hat{x}_1|^2 +|\tilde{x}_2-\hat{x}_2|^2 \bigr)^2} \frac{\delta_{b,\beta}(\hat{x}_1) \phi(\hat{x}_2)}{|\hat{x}_1-\hat{x}_2|} \, d\hat{x}_1d\hat{x}_2 \\ \label{Psi0bbeta}
 & = & \iint_{\mathbb{R}^3 \times \mathbb{R}^3} \frac{c_{12}}{|(\check{x}_1,\check{x}_2)|^4} \frac{\delta_{b,\beta}(\check{x}_1+\tilde{x}_2) \phi(\check{x}_2+\tilde{x}_2)}{|\check{x}_1-\check{x}_2|} \, d\check{x}_1d\check{x}_2 .
\end{eqnarray}
In the following lemma we give an analytic expression for (\ref{Psi0bbeta}) with $\phi(x)=e^{-x^2}$ at $b=0$, where the Coulomb singularity has its maximal effect according to (\ref{h1h2pde}).

\begin{lemma}
\label{lemmaPsi0beta}
For a test function $\phi(x)=e^{-x^2}$, the leading order asymptotic expression (\ref{Psi0bbeta}) at $b=0$, i.e.
\begin{equation}
 \Psi^{(0)}_{0,\beta}(\tilde{x}_2,\tilde{x}_2) := \iint_{\mathbb{R}^3 \times \mathbb{R}^3} \frac{c_{12}}{|(\check{x}_1,\check{x}_2)|^4} \frac{\delta_{0,\beta}(\check{x}_1+\tilde{x}_2) e^{-(\check{x}_2+\tilde{x}_2)^2}}{|\check{x}_1-\check{x}_2|} \, d\check{x}_1d\check{x}_2 ,
\label{Psi0beta}
\end{equation}
is given by
\begin{eqnarray}
\label{pcf}
 \Psi^{(0)}_{0,\beta}(\tilde{x}_2,\tilde{x}_2) & = & 2 c_{12} \sqrt{\pi \beta} |\tilde{x}_2|^{-2} e^{-(\beta+1) |\tilde{x}_2|^{2}} \int_0^{\frac{\pi}{2}} \left\{ \sqrt{a\pi} \left[ e^{\frac{1}{2} v^2(r)} -e^{\frac{1}{2} w^2(r)} \right] \right. \\ \nonumber & & \left. +\sqrt{\tfrac{a}{2}} \pi \left[ w(r) e^{\frac{1}{4} w^2(r)} V \bigl( -\tfrac{1}{2},w(r) \bigr) -v(r) e^{\frac{1}{4} v^2(r)} V \bigl( -\tfrac{1}{2},v(r) \bigr) \right] \right\} \cos(r) \, dr \\ \label{dawson}
 & = & 2 c_{12} \sqrt{\pi \beta} |\tilde{x}_2|^{-2} e^{-(\beta+1) |\tilde{x}_2|^{2}} \int_0^{\frac{\pi}{2}} \left\{ \sqrt{a\pi} \left[ e^{\frac{1}{2} v^2(r)} -e^{\frac{1}{2} w^2(r)} \right] \right. \\ \nonumber & & \left. +\sqrt{2\pi a} \left[ w(r) e^{\frac{1}{2} w^2(r)} F \bigl( w(r)/\sqrt{2} \bigr) -v(r) e^{\frac{1}{2} v^2(r)} F \bigl( v(r)/\sqrt{2} \bigr) \right] \right\} \cos(r) \, dr ,
\end{eqnarray}
with
\begin{equation}
 v(r) := \tfrac{2 \bigl( \beta \sin(r) -\cos(r) \bigr) |\tilde{x}_2|}{\sqrt{2 \bigl( \beta \sin^2(r) +\cos^2(r) \bigr)}}, \quad w(r) := \tfrac{2 \bigl( \beta \sin(r) +\cos(r) \bigr) |\tilde{x}_2|}{\sqrt{2 \bigl( \beta \sin^2(r) +\cos^2(r) \bigr)}} ,
\label{vwr}
\end{equation}
where $V\bigl( -\tfrac{1}{2},\cdot \bigr)$ denotes a parabolic cylinder function and $F(s) := e^{-s^2} \int_0^s e^{t^2} \, dt$ is the Dawson integral\footnote{The parabolic cylinder function and the Dawson integral are related by $V\bigl( -\tfrac{1}{2},t \bigr) =\frac{2}{\sqrt{\pi}} e^{\frac{1}{4} t^2} F(t/\sqrt{2})$.}, see \cite{AS}.
In the asymptotic limit $\tilde{x}_2 \rightarrow 0$ we get
\begin{equation}
 \Psi^{(0)}_{0,\infty}(\tilde{x}_2,\tilde{x}_2) := \lim_{\beta \rightarrow \infty} \Psi^{(0)}_{0,\beta}(\tilde{x}_2,\tilde{x}_2) \sim c_{12} \pi^2 |\tilde{x}_2|^{-2} e^{-|\tilde{x}_2|^{2}} .
\label{betainf}
\end{equation}
For a finite fixed value of $\beta$, $\Psi^{(0)}_{0,\beta}$ remains bounded along the diagonal, and we get  
\begin{equation}
 \lim_{\tilde{x}_2 \rightarrow 0} \Psi^{(0)}_{0,\beta}(\tilde{x}_2,\tilde{x}_2) = 2 c_{12} \pi^2 \beta .
\label{tildex0}
\end{equation}
\end{lemma}
\begin{proof}
To perform the calulation, we introduce hyperspherical coordinates for two particles given by
\begin{equation}
 x_1 =t \sin(r) \Phi(\theta_1,\phi_1), \quad x_2 =t \cos(r) \Phi(\theta_2,\phi_2) ,
\label{hc-two}
\end{equation}
with radial variable $t:= \sqrt{|x_1|^2 + |x_2|^2}$ and standard spherical coordinates. 
In these coordinates we have the $6d$ volume element
\[
 dV=t^5\sin^2(r) \cos^2(r) \sin(\theta_1) \sin(\theta_2) dt dr d\phi_1 d\phi_2 d\theta_1 d\theta_2 .
\] 
Thus (\ref{Psi0beta}) becomes
\begin{eqnarray*}
 \Psi^{(0)}_{0,\beta}(\tilde{x}_2,\tilde{x}_2) & = & 4\pi^2  c_{12} \int_{0}^{\infty} \int_{0}^{\frac{\pi}{2}} \int_{0}^{\pi} \int_{0}^{\pi} \frac{1}{t^4} \frac{1}{t\sin(r)} \left( \frac{\beta}{\pi} \right)^{\frac{3}{2}} e^{-\beta \bigl( t^2 \sin^2(r) +|\tilde{x}_2|^2 +2t\sin(r)|\tilde{x}_2| \cos(\theta_1) \bigr)} \\
 & & \times e^{-\bigl( t^2 \cos^2(r) +|\tilde{x}_2|^2 +2t\cos(r)|\tilde{x}_2| \cos(\theta_2) \bigr)} t^5\sin^2(r) \cos^2(r) \sin(\theta_1) \sin(\theta_2) dt dr d\theta_1 d\theta_2 \\
 & = & 4\pi^2 c_{12} \left( \frac{\beta}{\pi} \right)^{\frac{3}{2}} e^{-(\beta +1) |\tilde{x}_2|} \int_{0}^{\infty} \int_{0}^{\frac{\pi}{2}} e^{-t^2\bigl( \beta \sin^2(r) +\cos^2(r) \bigr)} \\
 &  & \times \frac{\sinh \bigl( 2\beta t\sin(r)|\tilde{x}_2|\bigr)}{\beta t\sin(r)|\tilde{x}_2|} \frac{\sinh \bigl( 2t\cos(r)|\tilde{x}_2|\bigr)}{t\cos(r)|\tilde{x}_2|} \sin(r) \cos^2(r) dt dr ,
\end{eqnarray*}
and, introducing the variables
\[
 a:= \beta \sin^2(r) +\cos^2(r), \quad b:= 2 \beta \sin(r) |\tilde{x}_2|, \quad d:= 2 \cos(r) |\tilde{x}_2| ,
\]
we have
\begin{eqnarray}
\label{Psi0betaint}
 \Psi^{(0)}_{0,\beta}(\tilde{x}_2,\tilde{x}_2) & = & 2 c_{12} \sqrt{\pi \beta} |\tilde{x}_2|^{-2} e^{-(\beta +1) |\tilde{x}_2|^2} \\ \nonumber
 & & \times \int_{0}^{\infty} \int_{0}^{\frac{\pi}{2}} e^{-at^2} \left[ \cosh \bigl( (b+d) t \bigr) -\cosh \bigl( (b-d) t \bigr) \right] t^{-2} \cos(r) dr dt .
\end{eqnarray}
The integral with respect to $t$ can be evaluated analytically. As a start, partial integration gives 
\begin{eqnarray}
\label{Icosh}
 \lefteqn{\int_{0}^{\infty} e^{-at^2} \left[ \cosh \bigl( (b+d) t \bigr) -\cosh \bigl( (b-d) t \bigr) \right] t^{-2} dt} \hspace{3cm} \\ \nonumber
 & = & \underbrace{\left. -e^{-at^2} \left[ \cosh \bigl( (b+d) t \bigr) -\cosh \bigl( (b-d) t \bigr) \right] t^{-1} \right|_{0}^{\infty}}_{=0} \\ \nonumber
 & & -2a \int_{0}^{\infty} e^{-at^2} \left[ \cosh \bigl( (b+d) t \bigr) -\cosh \bigl( (b-d) t \bigr) \right] dt \\
 & &  +\int_{0}^{\infty} e^{-at^2} \left[ (b+b) \sinh \bigl( (b+d) t \bigr) -(b-d)\sinh \bigl( (b-d) t \bigr) \right] t^{-1} dt .\quad\quad
\end{eqnarray}
Next we use the formulas, cf.~\cite{GR07} [p.~390, Eq.~3562 (1.,2.)], \cite{AS} [19.14.4] and \cite{AS} [19.3.8],
\begin{eqnarray*}
 \int_{0}^{\infty} e^{-at^2} \cosh(\gamma t) dt & = & \tfrac{1}{2} \Gamma(1) \tfrac{1}{\sqrt{2a}} e^{\frac{\gamma^2}{8a}} \left[ D_{-1} \bigl(-\tfrac{\gamma}{\sqrt{2a}} \bigr) +D_{-1} \bigl(\tfrac{\gamma}{\sqrt{2a}} \bigr) \right] \\
 & = & \tfrac{\pi}{2} \tfrac{1}{\sqrt{2a}} e^{\frac{\gamma^2}{8a}} V \bigl( \tfrac{1}{2}, \tfrac{\gamma}{\sqrt{2a}} \bigr) \\
 & = & \tfrac{1}{2} \sqrt{\tfrac{\pi}{a}} e^{\frac{\gamma^2}{4a}}
\end{eqnarray*}
\begin{eqnarray*}
 \int_{0}^{\infty} e^{-at^2} \sinh(\gamma t) t^{-1} dt & = & \lim_{\mu \rightarrow 0} \tfrac{1}{2} \Gamma(2\mu) e^{\frac{\gamma^2}{8a}} \left[ D_{-2\mu} \bigl(-\tfrac{\gamma}{\sqrt{2a}} \bigr) +\sin \bigl( (2\mu-\tfrac{1}{2}) \pi \bigr) D_{-2\mu} \bigl(\tfrac{\gamma}{\sqrt{2a}} \bigr) \right] \\
 & = & \tfrac{\pi}{2} e^{\frac{\gamma^2}{8a}} V \bigl( -\tfrac{1}{2}, \tfrac{\gamma}{\sqrt{2a}} \bigr) ,
\end{eqnarray*}
where $D_{-2\mu}$ and $V$ denote parabolic cylinder functions, and obtain for the integral (\ref{Icosh}) the expression
\begin{eqnarray}
\nonumber
 \lefteqn{\int_{0}^{\infty} e^{-at^2} \left[ \cosh \bigl( (b+d) t \bigr) -\cosh \bigl( (b-d) t \bigr) \right] t^{-2} dt} \hspace{3cm} \\ \nonumber
 & = & \sqrt{a\pi} e^{\frac{(b-d)^2}{4a}} -\sqrt{a\pi} e^{\frac{(b+d)^2}{4a}} \\ \nonumber
 & & +\tfrac{\pi}{2}(b+d) e^{\frac{(b+d)^2}{8a}} V \bigl( -\tfrac{1}{2}, \tfrac{b+d}{\sqrt{2a}} \bigr) -\tfrac{\pi}{2}(b-d) e^{\frac{(b-d)^2}{8a}} V \bigl( -\tfrac{1}{2}, \tfrac{b-d}{\sqrt{2a}} \bigr) \\ \label{IntPsi0beta}
 & = & \sqrt{a\pi} \left[ e^{\frac{1}{2} v^2} -e^{\frac{1}{2} w^2} \right] \\ \nonumber
 & & +\sqrt{\tfrac{a}{2}} \pi \left[ w e^{\frac{1}{4} w^2} V \bigl( -\tfrac{1}{2},w \bigr) -v e^{\frac{1}{4} v^2} V \bigl( -\tfrac{1}{2},v \bigr) \right] .
\end{eqnarray}
Inserting into (\ref{Psi0betaint}) yields (\ref{pcf}).

For our purposes, it is important to understand the behavior of (\ref{IntPsi0beta}) for $\beta \rightarrow \infty$.
Given a fixed value of $0<r<\frac{\pi}{2}$, we get
\[
 a \sim \beta \sin^2(r), \quad v \sim \sqrt{2\beta} |\tilde{x}_2|, \quad w \sim \sqrt{2\beta} |\tilde{x}_2|, \quad v-w \sim -\frac{2^{\frac{3}{2}}}{\sqrt{\beta}} \cot(r) |\tilde{x}_2| ,
\]
and using the asymptotic expansion \cite{AS} [19.8.2], for $v(r),w(r) \gg \frac{1}{2}$,
\[
 V \bigl( -\tfrac{1}{2}, x) \sim \sqrt{\tfrac{2}{\pi}} e^{\frac{1}{4} x^2} x^{-1} \bigl( 1 +x^{-2} + \cdots \bigr) ,
\]
we get in the limit $\beta \rightarrow \infty$ for (\ref{IntPsi0beta}) the asymptotic expression
\[
 \sim \sqrt{a\pi} \left[ e^{\frac{1}{2} w^2} w^{-2} -e^{\frac{1}{2} v^2} v^{-2} +\cdots \right] 
\]
Now  let us consider a Taylor expansion of (\ref{IntPsi0beta}) with respect to $v-w$, keeping only those terms that do not vanish in the limit $\beta \rightarrow \infty$, i.e.
\begin{equation}
 e^{\frac{1}{2} w^2} w^{-2} -e^{\frac{1}{2} v^2} v^{-2} +\cdots
 \sim -\left[ w^{-1} (v-w) +\tfrac{1}{2} (v-w)^2 +\tfrac{1}{3!} w (v-w)^3 +\cdots \right] e^{\frac{1}{2} w^2} ,
\label{Taylorvmw}
\end{equation}
where we only keep terms up to ${\cal O}(\beta^{-1})$. For $|\tilde{x}_2|$ close to zero, we consider  only  the leading order term, i.e.
\[
 \Psi^{(0)}_{0,\beta}(\tilde{x}_2,\tilde{x}_2) \sim 4 \pi c_{12} \sqrt{2\beta} |\tilde{x}_2|^{-1} e^{-(\beta +1) |\tilde{x}_2|^2} \int_{0}^{\frac{\pi}{2}} w^{-1}(r) e^{\frac{1}{2} w^2(r)} \cos^2(r) dr ,
\]
and get in the limit $\beta \rightarrow \infty$
\begin{eqnarray*}
 \Psi^{(0)}_{0,\infty}(\tilde{x}_2,\tilde{x}_2)
 & \sim & 4 \pi c_{12} |\tilde{x}_2|^{-2} e^{-|\tilde{x}_2|^2} \int_{0}^{\frac{\pi}{2}} \cos^2(r) dr \\
 & \sim & \pi^2 c_{12} |\tilde{x}_2|^{-2} e^{-|\tilde{x}_2|^2} ,
\end{eqnarray*}
which corresponds to the asymptotic expression (\ref{betainf}). To illustrate the behavior of the Taylor expansion (\ref{Taylorvmw}), we have plotted in Fig.~\ref{fig1} the function $\Psi^{(0)}_{0,\beta}$  for $\beta=1000$, and approximations of $\Psi^{(0)}_{0,\infty}$ using Taylor polynomials (\ref{Taylorvmw}) of first, second, and third order, respectively. 

Finally, we consider the limit $|\tilde{x}_2| \rightarrow 0$ for large but finite values of $\beta$, i.e., $v,w \ll \frac{1}{2}$, using the Taylor approximation, cf.~\cite{AS} [19.3.6],
\[
 V\bigl( -\tfrac{1}{2}, x \bigr) \approx V\bigl( -\tfrac{1}{2}, 0 \bigr) +V'\bigl( -\tfrac{1}{2}, 0 \bigr) x = \sqrt{\tfrac{2}{\pi}} x ,
\]
which gives
\begin{multline*}
 \sqrt{a\pi} \left[ e^{\frac{1}{2} v^2(r)} -e^{\frac{1}{2} w^2(r)} \right] +\sqrt{\tfrac{a}{2}} \pi \left[ w(r) e^{\frac{1}{4} w^2(r)} V \bigl( -\tfrac{1}{2},w(r) \bigr) -v(r) e^{\frac{1}{4} v^2(r)} V \bigl( -\tfrac{1}{2},v(r) \bigr) \right] \\ \approx \tfrac{1}{2} \sqrt{a\pi} \bigl( w^2(r)-v^2(r) \bigr) .
\end{multline*}
With
\[
 w^2(r)-v^2(r)= 8 \beta |\tilde{x}_2| \frac{\sin(r)\cos(r)}{\beta \sin^2(r)+\cos^2(r)} ,
\]
the right hand side can be further approximated by
\[
 \tfrac{1}{2} \sqrt{a\pi} \bigl( w^2(r)-v^2(r) \bigr) \approx 4 \sqrt{\pi \beta} |\tilde{x}_2|^2 \cos(r) ,
\]
for fixed $r>0$ and sufficiently large $\beta$, which gives the limit (\ref{tildex0}). Roughly speaking, for finite values of $\beta$, the function $\Psi^{(0)}_{0,\beta}$ can be viewed as a regularization of the singular function $\Psi^{(0)}_{0,\infty}$. This behavior is illustrated in Fig.~\ref{fig1}, where $\Psi^{(0)}_{0,\beta}$, for $\beta=1000$, $2000$ and $3000$, has been compared with $\Psi^{(0)}_{0,\infty}$, obtained from the third-order Taylor polynomial (\ref{Taylorvmw}).
\end{proof}

\begin{figure}[p]
\begin{center}
\includegraphics[scale=0.7]{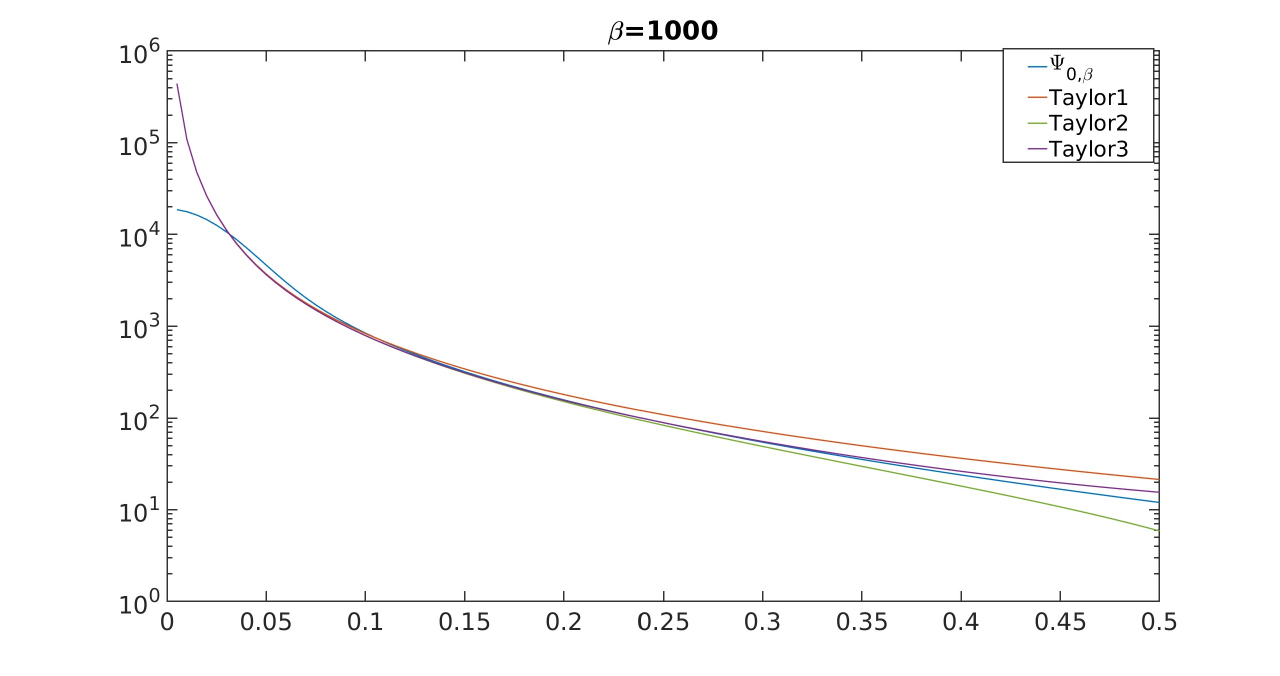} \vspace{0.3cm}
\includegraphics[scale=0.7]{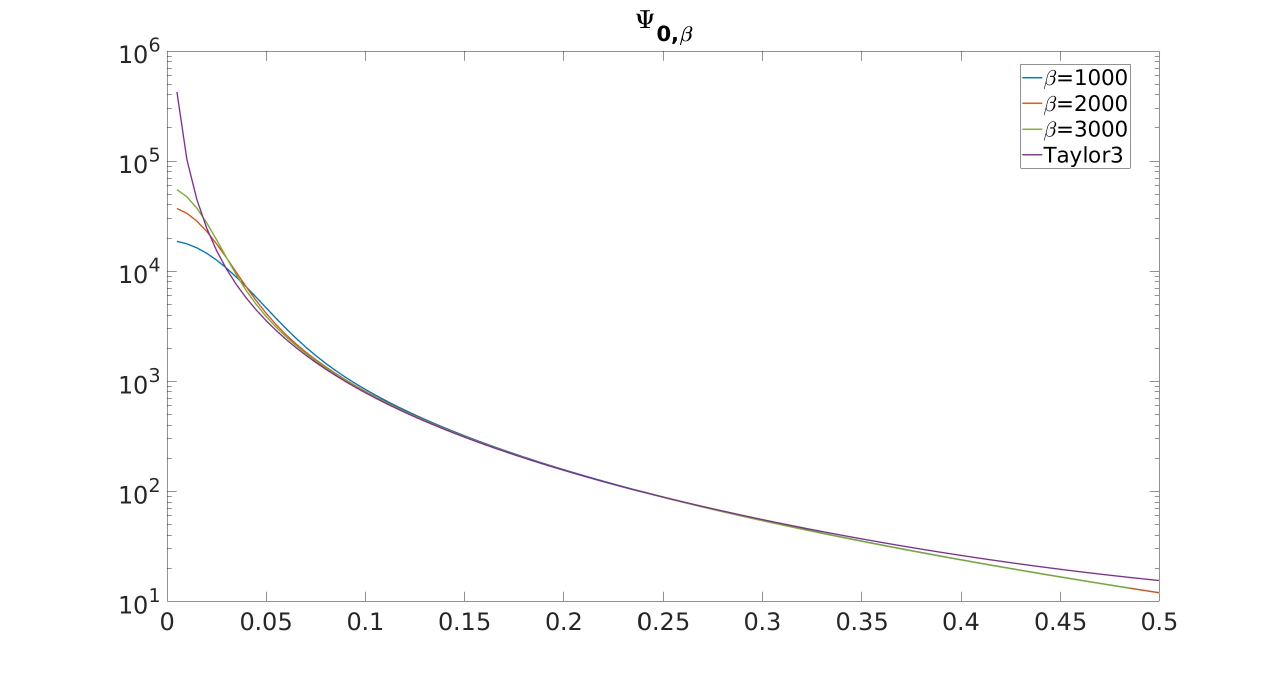}
\end{center}
\vspace{0.5cm}
\caption{
(top) The function $\Psi^{(0)}_{0,\beta}$ for $\beta=1000$, and approximations of $\Psi^{(0)}_{0,\infty}$ using Taylor polynomials (\ref{Taylorvmw}) of first, second, and third order, respectively.
(bottom) The functions $\Psi^{(0)}_{0,\beta}$, for $\beta=1000$, $2000$ and $3000$, compared to $\Psi^{(0)}_{0,\infty}$, obtained from the third order Taylor polynomial (\ref{Taylorvmw}).
}
\label{fig1}
\end{figure}

The previous lemma shows that it is possible to apply Lebesgue's dominated convergence theorem   in (\ref{I1ab}), i.e.
\begin{eqnarray*}
 I^{(0)}_1(a,0) & := & \lim_{\beta \rightarrow \infty} 4 \iint_{\mathbb{R}^3 \times \mathbb{R}^3} \frac{\delta_{a,\beta}(\tilde{x}_1) \phi(2\tilde{x}_2-\tilde{x}_1)}{|\tilde{x}_1-\tilde{x}_2|} \Psi^{(0)}_{0,\beta}(\tilde{x}_2,\tilde{x}_2) \, d\tilde{x}_1 d\tilde{x}_2 \\
 & = & 4 \iint_{\mathbb{R}^3 \times \mathbb{R}^3} \frac{\delta_{a}(\tilde{x}_1) \phi(2\tilde{x}_2-\tilde{x}_1)}{|\tilde{x}_1-\tilde{x}_2|} \Psi^{(0)}_{0,\infty}(\tilde{x}_2,\tilde{x}_2) \, d\tilde{x}_1 d\tilde{x}_2 \\
 & = & 4 \int_{\mathbb{R}^3} \frac{\phi(2\tilde{x}_2-a)}{|a-\tilde{x}_2|} \Psi^{(0)}_{0,\infty}(\tilde{x}_2,\tilde{x}_2) \, d\tilde{x}_1 .
\end{eqnarray*}
From (\ref{betainf}) in Lemma \ref{lemmaPsi0beta} we can derive the leading order singular asymptotic term by taking
\[
 \phi(2\tilde{x}_2-a) \Psi^{(0)}_{0,\infty}(\tilde{x}_2,\tilde{x}_2) \sim c_{12} \pi^2 |\tilde{x}_2|^{-2} e^{-|\tilde{x}_2|^{2}} ,
\]
and using Laplace's expansion of the Coulomb potential, i.e.
\begin{equation}
 \frac{1}{|a-\tilde{x}_2|} = \sum_{\ell=0}^{\infty} \frac{r_{<}^{\ell}}{r_{>}^{\ell+1}}
 \frac{4\pi}{2\ell+1}
 \sum_{m=-\ell}^{\ell} (-1)^{m} Y_{\ell,m}(\theta,\varphi) Y_{\ell,-m}(\tilde{\theta},\tilde{\varphi}) ,
\label{Laplacexp}
\end{equation}
with $r_{<} := \min \{|a|,|\tilde{x}_2|\}$ and $r_{>} := \max \{|a|,|\tilde{x}_2|\}$. This gives
\[
 I^{(0)}_1(a,0) \sim c_{12} 4\pi^3 \frac{1}{|a|} \int_{0}^{|a|} e^{-|\tilde{x}_2|^{2}} d|\tilde{x}_2| + c_{12} 4\pi^3 \int_{|a|}^{\infty} \frac{1}{|\tilde{x}_2|} e^{-|\tilde{x}_2|^{2}} d|\tilde{x}_2| .
\]
The two integrals correspond to transcendental functions with known asymptotic behavior, see \cite{AS}~[7.1.5,5.1.10], i.e.
\begin{equation}
 \frac{1}{|a|} \int_{0}^{|a|} e^{-|\tilde{x}_2|^{2}} d|\tilde{x}_2| = \frac{\sqrt{\pi}}{2|a|} \mbox{erf}(|a|) \sim 1 -\tfrac{1}{3} |a|^2 \cdots ,
\label{erfa}
\end{equation}
and
\begin{equation}
 \int_{|a|}^{\infty} \frac{1}{|\tilde{x}_2|} e^{-|\tilde{x}_2|^{2}} d|\tilde{x}_2| = \frac{1}{2} \int_{|a|^2}^{\infty} t^{-1} e^{-t} dt = \frac{1}{2} \mbox{E}_1(|a|^2) \sim -\frac{1}{2} \gamma -\ln(|a|) +\tfrac{1}{2} |a| -\tfrac{1}{8} |a|^2 +\cdots ,
\label{E1}
\end{equation}
where $\mbox{erf}$ and $\mbox{E}_1$ denote the error function and the exponential integral, respectively.
The integral (\ref{erfa}) is an analytic function in $|a|^2$ and therefore does not contribute to the singular asymptotic behavior, while the integral (\ref{E1}) has in leading order a singular term $\ln(|a|)$ in its asymptotic expansion.\footnote{Here $\gamma$ denotes the Euler-Mascheroni constant.}

The next term in the asymptotic expansion (\ref{Kbeta-asymp}) is given by
\begin{eqnarray*}
 I_2(a,b) & := & \lim_{\beta \rightarrow \infty} \frac{1}{2} \iint_{\mathbb{R}^3 \times \mathbb{R}^3} \delta_{a,\beta}(x_1) \phi(x_2) \phi \bigl( \tfrac{1}{2} (x_1+x_2) \bigr) \delta_{b,\beta} \bigl( \tfrac{1}{2} (x_1+x_2) \bigr) \, dx_1dx_2 \\
 & = & \lim_{\beta \rightarrow \infty} 4 \iint_{\mathbb{R}^3 \times \mathbb{R}^3} \delta_{a,\beta}(\tilde{x}_1) \phi(2\tilde{x}_2-\tilde{x}_1) \phi(\tilde{x}_2) \delta_{b,\beta}(\tilde{x}_2) \, d\tilde{x}_1d\tilde{x}_2 \\
 & = & 4 \phi(2b-a) \phi(b) ,
\end{eqnarray*}
and does not contribute to the singular asymptotics.

Let us comment briefly on the modification required for the exchange part of the $2p1h$ Green's function. The corresponding integral (\ref{Kbeta}) for the exchange part becomes 
\[
  K_{\beta}(a,b) := \lim_{\beta \rightarrow \infty} \iint_{\mathbb{R}^3 \times \mathbb{R}^3} \delta_{a,\beta}(x_1) \phi(x_2) \bigl( \hat{H}^{(0)}(\omega) \phi \otimes \delta_{b,\beta} \bigr) (x_1,x_2) \, dx_1dx_2 ,
\]
and in the asymptotic limit we get
\begin{equation}
  K_{\beta}(a,b) \sim \iint_{\mathbb{R}^3 \times \mathbb{R}^3} \iint_{\mathbb{R}^3 \times \mathbb{R}^3} \frac{\delta_{a,\beta}(x_1) \phi(x_2) \delta_{b,\beta}(x_3) \phi(x_4)}{|x_1-x_2|\bigl( |x_1-x_4|^2 +|x_2-x_3|^2 \bigr)^2|x_3-x_4|} \, dx_1dx_2dx_3dx_4 .
\label{Kbeta-ex}
\end{equation}
In the asymptotic expansion (\ref{P0tildePsi}) we consider only leading order contributions with zero relative angular momentum due to the presence of the projection operator $P_0$, which means that the corresponding asymptotic terms are symmetric with respect to a particle exchange. Therefore, the same reasoning applies to the exchange part and we get the same leading order term in the asymptotic expansion of (\ref{Kbeta-ex}). This leads to spin-dependent cancellations between the direct and exchange terms, similar to the so-called self-interaction correction in the Fock operator of the Hartree-Fock model. Higher order direct and exchange terms also have contributions from non-zero relative angular momenta and hence possesses different asymptotics. 

Finally, we summarize our discussion in the following lemma.

\begin{lemma}
\label{lemmaKab}
The kernel functions (\ref{sumjHjd}) and (\ref{sumjHje}), which belong to the partially contracted operator (\ref{Hj}) and its simplified variant (\ref{H0j}), have a singular asymptotic behavior, i.e., they are, modulo smooth terms, of the form
\[
 H_{d(x)}(x_1,x_2,\omega) \sim a_0(x_1,\omega) \ln(|x_1-x_2|) +a_1(x_1,\omega) |x_1-x_2| + \cdots ,
\]
where only terms with zero relative angular momentum are considered. Such kernel functions belong to the classical pseudo-differential operators of the H\"ormander class $S^{-3}(\mathbb{R}^3,\mathbb{R}^3)$.
\end{lemma}

\section{Asymptotic smoothness of Feynman diagrams and outlook}
\label{asymptoticsmoothness}
So far, we determined the asymptotic behavior of one-particle Green's functions and the singular behavior of partially contracted kernel functions. Now, let us shortly discuss how this approach can be applied to Feynman diagrams to give a new refined classification system for them that contains besides the order of perturbation also the asymptotic smoothness. In the future, this may be used to design new sparse grid combination methods \cite{GSZ92,GH2014, RG, GS24} in a multiscale fashion with improved computational complexities. 

Conventionally, it is common practice to classify Feynman diagrams according to their number of interaction lines, corresponding to the order of perturbation theory in which they appear. Such a classification scheme roughly corresponds to the computational complexity of numerical simulations for the diagramms. Now, following our previous discussion, we refine the classification of Feynman diagrams by introducing the asymptotic smoothness as an additional order parameter. For this purpose, we decompose a Feynman diagram with two external lines into a singular part with a fast decay at infinity and a smooth remainder, i.e.
\begin{equation}
 F(x,\tilde{x},\omega) = F_s(x,\tilde{x},\omega) +F_{\infty}(x,\tilde{x},\omega) ,
\label{Fdecomp}
\end{equation}
with $F_{\infty} \in C^{\infty}(\mathbb{R}^3\times \mathbb{R}^3 \times \mathbb{R})$ such that the singular part $F_s$ satisfies the following definition.

\begin{definition}
	A Feynman diagram $F(x,\tilde{x},\omega)$ with two external lines $x$ and $\tilde{x}$, where $\omega \in \mathbb{R}$ is considered as parameter, has asymptotic smoothness of order $p$ if it belongs to $C^{\infty}(\mathbb{R}^3\times \mathbb{R}^3 \setminus \{ 0 \})$, for any $\omega \in \mathbb{R}$, and its singular part $F_s$ satisfies the asymptotic smoothness property\footnote{In this definition, the fast asymptotic decay of $F_s$ at large distances, i.e.~for $|x-\tilde{x}| \rightarrow \infty|$, is triggered by the parameter $N$.}
\begin{equation}
\left| \partial_{x}^{\alpha} \partial_{\tilde{x}}^{\beta}
F_s(x,\tilde{x},\omega) \right| \leq C_{\alpha,\beta,N,\omega} |x-\tilde{x}|^{-3-p-|\alpha|-|\beta|-N}
\label{asympestimate}
\end{equation}
for $-3-p-|\alpha|-|\beta|-N  <0$, and each $N \in \mathbb{N}_0$,
where it has bounded partial derivatives for $|\alpha|+|\beta| \leq -3-p$. Note that the constant $C_{\alpha,\beta,N,\omega} >0$  in (\ref{asympestimate}) may depend on $\omega$. 
\label{def1}
\end{definition}

Such a decomposition can be achieved, for example, by using an appropriate cutoff function perpendicular to the diagonal.     
The previous definition is motivated by the corresponding property of kernel functions of pseudo-differential operators of order $p<0$, see~\cite{Stein} for further details. 

To illustrate the benefit of our two-parameter classification scheme for Feynman diagrams, consider the one-particle Green's function discussed in the previous sections. For the corresponding Feynman diagrams, we observe a correlation  between their order in perturbation theory, i.e.~the number of interaction lines, and their asymptotic smoothness, which is schematically shown in Fig.~\ref{fig4}. Regarding the asymptotic smoothness of Feynman diagrams in higher orders of perturbation theory, our results are still incomplete. However, we have at least outlined a general approach that can be extended to higher orders. 
\begin{figure}[t]
\begin{center}
\includegraphics[scale=0.23]{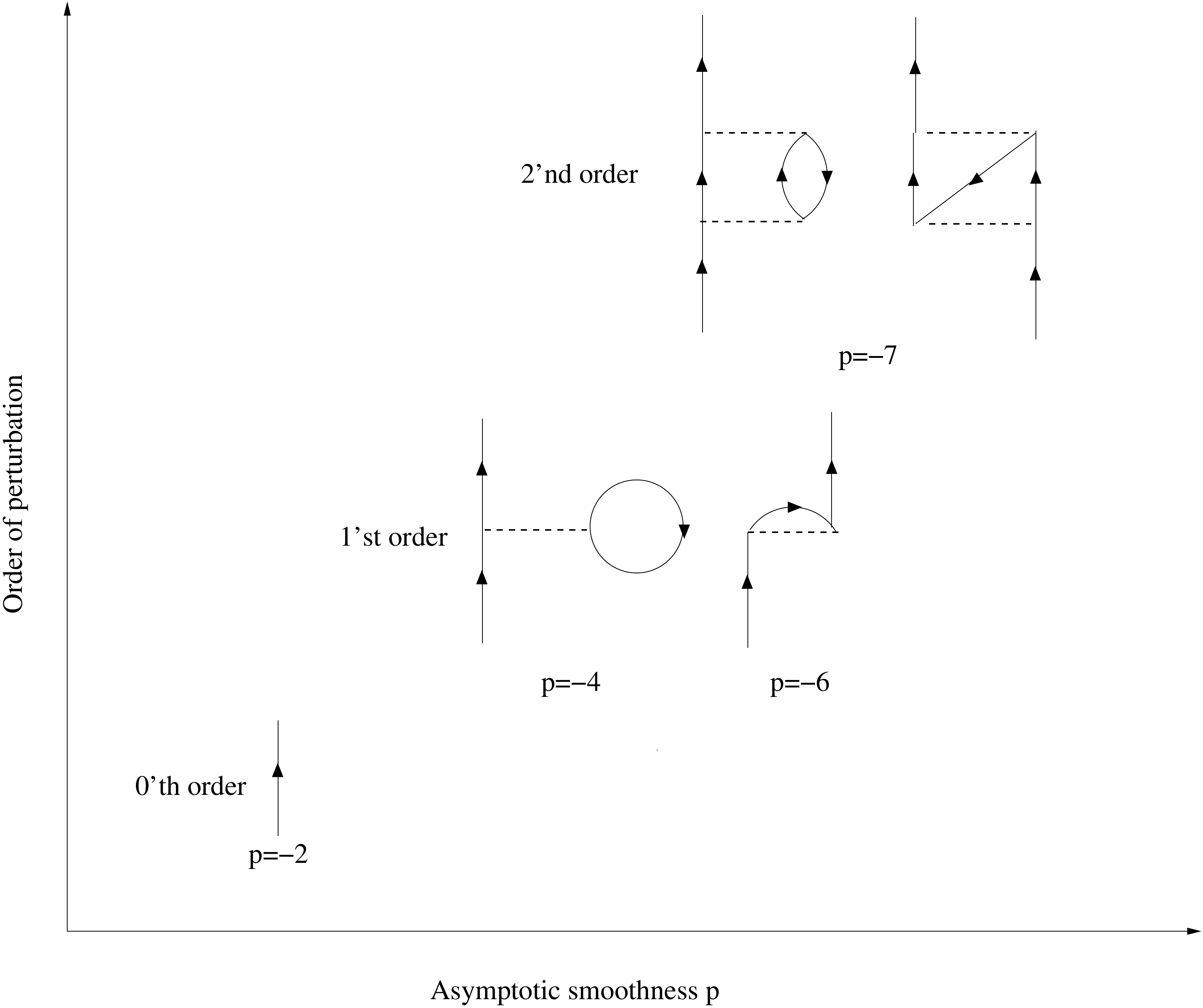}
\end{center}
\vspace{0.5cm}
\caption{Asymptotic smoothness versus order of perturbation theory for the singular part $F_s$ of certain Feynman diagrams contributing to the one-particle Green's function. The asymptotic smoothness parameter $p$ was given in Definition \ref{def1}.}
\label{fig4}
\end{figure}

Now, to take advantage of such a correlation, we suggest to design a sparse grid combination method along the lines of \cite{GH2014,GSZ92, GS24,RG} that correlates asymptotic smoothness with a hierarchical tensor product approximation scheme. In particular, the order of asymptotic smoothness of a Feynman diagram entails its approximation properties\footnote{Indeed it determines the Besov and mixed Sobolev regularity of a Feynman diagram.} with respect to hyperbolic cross approximation schemes in tensor product bases, cf.~\cite{FHS07,FFH21,Yserentant} for further details.
Thus, instead of performing a single numerical calculation at the highest order of perturbation theory and with the largest possible basis set, the combination technique decomposes the calculation into a telescopic sequence and restricts actual calculations to a hyperbolic cross/sparse grid with respect to perturbation order and basis set size. Such an approach can also be viewed as a two-variate extrapolation method where one direction corresponds to the order of perturbation and the other direction corresponds to the asymptotic smoothness. The development, the error analysis, and the cost complexity estimation of such a multiscale approach will be future work.

In conclusion let us make the following remarks: Despite good progress in the accuracy and efficient implementation of many-particle models, their applicability is severely limited by the presence of singularities. While there is a solid knowledge of the effects of singularities for wavefunction-based methods, much less is known for reduced quantities such as Green's and response functions. This is partly due to the fact that these quantities are not easily accessible via the Schr\"odinger equation. Instead, their very definitions involve concepts of many-particle theory, such as second quantization, which are difficult to put into a rigorous mathematical framework. The present work is a first attempt to unveil the singular structure of dynamical reduced quantities by means of a case study for the one-particle Green's function. To this end, we applied techniques from singular analysis, which turned out to be useful for our purposes. Besides explicit asymptotic expansions concerning the small distance behavior of the Green's function, we derived a refined classification of the corresponding Feynman diagrams, taking into account their asymptotic smoothness near the diagonal.  
Our study of the small distance behavior of Green's functions is reminiscent of related work in quantum field theory.
While a proper treatment of singularities via regularization and renormalization is crucial there, this issue seems quite irrelevant in electronic structure theory.    
However, a more careful consideration reveals the tight interplay between singular structures and computational complexity.
The extended classification scheme discussed above provides a first step in this direction and
the computational complexity of numerical algorithms for many-particle models may get improved by taking into account their singular structure. 

\section*{Acknowledgments}
The authors HJF and MG were supported by the \emph{Hausdorff Center for Mathematics} (HCM) in Bonn, funded by the Deutsche Forschungsgemeinschaft (DFG, German Research Foundation) under Germany's Excellence Strategy -- EXC-2047/1 -- 390685813 and by the CRC 1060 \emph{The Mathematics of Emergent Effects} -- 211504053 of the Deutsche Forschungsgemeinschaft.

\vspace{1cm}
\appendix
\noindent {\bf \Large Appendix}

\section{Partial contractions of kernel functions}
\label{partialcontraction}
The calculation of the integral, 
\[
 I_{\beta}(a,b) :=\int_{\mathbb{R}^3} \frac{\delta_{a,\beta}(x_1,x_2) \otimes \phi(x_3) \cdot \delta_{a,\beta}(\tilde{x}_1,\tilde{x}_2) \otimes \phi(\tilde{x}_3)}{|x-\tilde{x}|} dxd\tilde{x} ,
\]
can be done in several steps. Using (\ref{CG}) and the Gaussian product formula, cf.~\cite{HJO},
\[
 e^{-\alpha|x-a|^2} e^{-\beta|x-b|^2} = e^{-\frac{\alpha \beta}{\alpha+\beta} |a-b|^2} e^{-(\alpha+\beta) \left|x-\frac{\alpha a + \beta b}{\alpha+\beta} \right|^2} ,
\]  
we get in a first step
\[
 \int_{\mathbb{R}^2} e^{-\beta|x_1-a|^2} e^{-|x_1-x_2|^2t^2} \, dx_1 = \tfrac{\pi}{\beta+t^2} e^{-\frac{\beta t^2}{\beta+t^2} |x_2-a|^2} ,
\]
and in a second step
\[
 \int_{\mathbb{R}^2} e^{-\frac{\beta t^2}{\beta+t^2} |x_2-a|^2} e^{-\beta|x_2-b|^2} \, dx_2 = \tfrac{\pi}{\beta+\frac{\beta t^2}{\beta+t^2}} e^{-\frac{\beta t^2}{\beta+2t^2} |a-b|^2} .
\]
These integrals, together with (\ref{It}),  yield
\[
 I_{\beta}(a,b) = \int_{-\infty}^{\infty} \sqrt{\tfrac{\pi}{1+2t^2}} \tfrac{\beta}{\beta+2t^2} e^{-\frac{\beta t^2}{\beta+2t^2} |a-b|^2} \, dt
\] 
Now, substituting 
\[
 u^2=\frac{\beta t^2}{\beta+2t^2} ,
\]
for a sufficiently large $\beta$, we obtain 
\[
 I_{\beta}(a,b) = 2 \int_0^{\sqrt{\frac{\beta}{2}}} \sqrt{\tfrac{\pi}{1+2\left (1-\frac{1}{\beta} \right)u^2}} e^{-u^2|a-b|^2} .
\]
Finally, taking the limit $\beta \rightarrow \infty$, we get 
\begin{eqnarray*}
 \lim_{\beta \rightarrow \infty} I_{\beta}(a,b) & = & \lim_{\beta \rightarrow \infty} \left[ 2 \int_0^{\infty} \sqrt{\tfrac{\pi}{1+2\left (1-\frac{1}{\beta} \right)u^2}} e^{-u^2|a-b|^2} \, du - 2 \int_{\sqrt{\frac{\beta}{2}}}^{\infty} \sqrt{\tfrac{\pi}{1+2\left (1-\frac{1}{\beta} \right)u^2}} e^{-u^2|a-b|^2} \, du \right] \\
 & = & 2 \int_0^{\infty} \sqrt{\tfrac{\pi}{1+2u^2}} e^{-u^2|a-b|^2} \, du \\
  & = & \sqrt{\frac{\pi}{2}} e^{\frac{|a-b|^2}{4}} K_0 \left( \frac{|a-b|^2}{4} \right) ,
\end{eqnarray*}
and, up to a prefactor, we just recovered (\ref{K2d}).

\end{document}